\documentclass[10pt,letterpaper]{article}
\usepackage[top=0.85in,left=2.75in,footskip=0.75in]{geometry}

\usepackage{graphicx,amsmath,amssymb,bm}
\usepackage{amsthm}
\usepackage{algorithm,algorithmic}

\usepackage{changepage}

\usepackage{textcomp,marvosym}

\usepackage{cite}

\usepackage{nameref,hyperref}

\usepackage[right]{lineno}

\usepackage[nopatch=eqnum]{microtype}
\DisableLigatures[f]{encoding = *, family = * }

\usepackage[table]{xcolor}

\usepackage{array}

\newcolumntype{+}{!{\vrule width 2pt}}

\newlength\savedwidth

\newcommand\thickhline{\noalign{\global\savedwidth\arrayrulewidth\global\arrayrulewidth 2pt}%
\hline
\noalign{\global\arrayrulewidth\savedwidth}}

\raggedright
\usepackage[aboveskip=1pt,labelfont=bf,labelsep=period,justification=raggedright,singlelinecheck=off]{caption}

\makeatletter
\renewcommand{\@biblabel}[1]{\quad#1.}
\makeatother

\usepackage{lastpage,fancyhdr,graphicx}
\usepackage{epstopdf}
\fancyheadoffset[L]{2.25in}
\fancyfootoffset[L]{2.25in}
\newtheorem{definition}{Definition}
\newtheorem{theorem}{Theorem}
\newtheorem{proposition}{Proposition}

\begin{document}
\vspace*{0.2in}

\begin{flushleft}
{\Large
\textbf\newline{Autocratic strategies in Cournot oligopoly game} 
}
\newline
\\
Masahiko Ueda\textsuperscript{1*},
Shoma Yagi\textsuperscript{2},
Genki Ichinose\textsuperscript{2\dag}
\\
\bigskip
\textbf{1} Graduate School of Sciences and Technology for Innovation, Yamaguchi University, Yamaguchi, Japan
\\
\textbf{2} Department of Mathematical and Systems Engineering, Shizuoka University, Hamamatsu, Japan
\\
\bigskip

%
%





* m.ueda@yamaguchi-u.ac.jp

\dag ichinose.genki@shizuoka.ac.jp

\end{flushleft}
\section*{Abstract}
An oligopoly is a market in which the price of goods is controlled by a few firms.
Cournot introduced the simplest game-theoretic model of oligopoly, where profit-maximizing behavior of each firm results in market failure.
Furthermore, when the Cournot oligopoly game is infinitely repeated, firms can tacitly collude to monopolize the market.
Such tacit collusion is realized by the same mechanism as direct reciprocity in the repeated prisoner's dilemma game, where mutual cooperation can be realized whereas defection is favorable for both prisoners in a one-shot game.
Recently, in the repeated prisoner's dilemma game, a class of strategies called zero-determinant strategies attracts much attention in the context of direct reciprocity.
Zero-determinant strategies are autocratic strategies which unilaterally control payoffs of players by enforcing linear relationships between payoffs.
There were many attempts to find zero-determinant strategies in other games and to extend them so as to apply them to broader situations.
In this paper, first, we show that zero-determinant strategies exist even in the repeated Cournot oligopoly game, and that they are quite different from those in the repeated prisoner's dilemma game.
Especially, we prove that a fair zero-determinant strategy exists, which is guaranteed to obtain the average payoff of the opponents.
Second, we numerically show that the fair zero-determinant strategy can be used to promote collusion when it is used against an adaptively learning player, whereas it cannot promote collusion when it is used against two adaptively learning players.
Our findings elucidate some negative impact of zero-determinant strategies in the oligopoly market.

\section*{Author summary}
Repeated games have been used to analyze the rational decision-making of multiple agents in a long-term interdependent relationship.
Recently, a class of autocratic strategies, called zero-determinant strategies, was discovered in repeated games, which unilaterally controls payoffs of players via enforcing linear relations between payoffs.
So far, properties of zero-determinant strategies in social dilemma situations have extensively been investigated, and it has been shown that some zero-determinant strategies promote cooperation.
Moreover, zero-determinant strategies have been found in several games.
However, it has not been known whether zero-determinant strategies exist in oligopoly games.
In this paper, we investigate zero-determinant strategies in the repeated Cournot oligopoly game, which is the simplest mathematical model of oligopoly.
We prove the existence of zero-determinant strategies which unilaterally enforce linear relations between the payoff of the player and the average payoff of the opponents.
Furthermore, we numerically show that some zero-determinant strategy can promote collusion in a duopoly case, although it cannot promote collusion in a triopoly case.
Our results imply that zero-determinant strategies can be used to promote cooperation between firms even in an oligopoly market.


\section{Introduction}
\label{sec:introduction}
Oligopoly is one of the simplest game-theoretic situations in economics, where the price of goods is controlled by a few firms.
Cournot introduced a simple model of oligopoly, where multiple firms produce the same goods and the good's price decreases as the total amount of the goods increases \cite{FudTir1991,OsbRub1994}.
Each firm wants to maximize its profit, but needs to choose its production while considering production of the other firms.
When all firms are rational, the Cournot-Nash equilibrium is realized, where production of each firm is the best response against production of the other firms.
Because total production in this equilibrium is smaller than that in perfect competition, oligopoly results in market failure.

When the Cournot oligopoly game is infinitely repeated, the situation becomes worse for consumers.
Since firms become taking total profits obtained in future into account, firms can tacitly collude to monopolize the market \cite{Gib1992}, even if the Cournot-Nash equilibrium production is the only rational behavior of each firm in a one-shot game.
This is an example of the folk theorem in repeated games.
This result is similar to direct reciprocity in the repeated prisoner's dilemma game \cite{Now2006}.
In the prisoner's dilemma game, two prisoners independently choose cooperation or defection.
Defection is the best action in terms of the payoff of each prisoner regardless of the action of the other prisoner, and mutual defection is realized as a result of rational behavior.
Because mutual cooperation improves the payoffs of both prisoners in mutual defection, this game describes a kind of social dilemmas.
However, when the game is infinitely repeated, mutual cooperation can be realized as a result of long-term perspectives, similarly to tacit collusion in the repeated Cournot oligopoly game.

In 2012, a novel class of strategies, called zero-determinant (ZD) strategies, was discovered in the repeated prisoner's dilemma game \cite{PreDys2012}.
Counterintuitively, ZD strategies unilaterally control payoffs of players by enforcing linear relations between payoffs.
ZD strategies can be regarded as generalization of the equalizer strategy \cite{BNS1997}, which unilaterally sets the payoff of the opponent.
ZD strategies also contain the extortionate ZD strategy, by which a player never obtains payoff lower than that of the opponent, and the generous ZD strategy \cite{StePlo2013}, by which a player never obtains payoff higher than that of the opponent but promotes cooperation.
So far, ZD strategies were discovered in the prisoner's dilemma game \cite{PreDys2012,StePlo2013,Aki2016}, the public goods game \cite{HWTN2014,PHRT2015}, the continuous donation game \cite{McAHau2016}, the asymmetric prisoner's dilemma game \cite{TahGho2020,KTZ2024}, two-player potential games \cite{Ued2022,Ued2025}, and symmetric games with no generalized rock-paper-scissors cycles \cite{Ued2022b,Ued2023}.
Particularly, ZD strategies in games with infinite action sets are called autocratic strategies \cite{McAHau2016}.
Furthermore, ZD strategies were extended to broader situations, such as games with a discount factor \cite{HTS2015,McAHau2016,IchMas2018,GovCao2020}, games with observation errors \cite{HRZ2015,MamIch2020,MMI2021}, asynchronous games \cite{McAHau2017}, and stochastic games \cite{DRWW2021,LiuWu2022}.
Roles of ZD strategies in the context of evolution of cooperation have been extensively investigated \cite{HNS2013,AdaHin2013,StePlo2013,HNT2013,SzoPer2014,HWTN2015,Aki2016,CWF2022,CheFu2023}.

Although the Cournot oligopoly game contains some social dilemma structure as noted above, it is not completely the same as social dilemmas.
This is because the action sets in the Cournot oligopoly game are infinite sets, whereas the action sets in typical social dilemma games contain only two actions, that is, cooperation and defection.
Probably for that reason, ZD strategies have not been discovered in the Cournot oligopoly game, and it is not clear whether players can autocratically control payoffs in the game.

If such unilateral payoff control is possible in the Cournot oligopoly game, a firm can lead the other firms to collusion.
In fact, such control is possible by ZD strategies in the repeated prisoner's dilemma game \cite{PreDys2012}.
If the payoffs of two players are positively correlated in a linear relation enforced by a ZD strategy, increase in payoff of the opponent implies increase in payoff of a player using the ZD strategy.
Therefore, if the opponent gradually improves its payoff by learning, the payoff of the player using the ZD strategy also increases up to some local optima.
Because ZD strategies do not assume rationality of other players, they are useful even if some players are boundedly rational \cite{BisNai2000}, in contrast to traditional strategies in classical game theory \cite{HTS2015}.

In this paper, we make two key contributions.
First, we prove that ZD strategies (or, autocratic strategies) also exist in the repeated Cournot oligopoly game.
Concretely, we show that ZD strategies pinning their own payoffs, positively correlated ZD strategies \cite{CheZin2014,MMI2022}, which are the class containing the extortionate ZD strategy and the generous ZD strategy, and negatively correlated ZD strategies exist in the game.
Especially, we find that a fair ZD strategy exists in the game, by which a player unilaterally obtains the average payoff of the opponents.
Second, we specify performance of the fair ZD strategy in evolution of collusion.
We numerically investigate whether collusion is achieved when the fair ZD strategy is used against adaptively learning players \cite{PreDys2012,CheZin2014,MMI2022}.
We find that, in a two-player case, the fair ZD strategy promotes collusion of a learning player, but in a three-player case, it cannot promote collusion.

This paper is organized as follows.
In Section \ref{sec:model}, we introduce the Cournot oligopoly game and its repeated version.
In Section \ref{sec:preliminaries}, we introduce autocratic strategies, and explain their properties related to the existence.
In Section \ref{sec:results}, we prove that some types of autocratic strategies exist in the repeated Cournot oligopoly game.
In Section \ref{sec:numerical}, we provide numerical results.
Section \ref{sec:discussion} is devoted to discussion.
The proofs of all theoretical results are provided in Section \ref{sec:methods}.

\section{Model}
\label{sec:model}
We consider the $N$-player Cournot oligopoly game with $N\geq 2$ \cite{FudTir1991,OsbRub1994}.
In this game, multiple firms produce the same goods at the same cost.
The model assumes that the good's price is determined by consumer demand, and the price decreases as the total amount of the goods increases.
Each firm needs to choose its production while considering production of the other firms.
If the production of each firm is too small, it obtains small positive profit.
However, if the production of a firm is too large, all firms may obtain negative profit.
Therefore, the situation is game-theoretic.
In the simplest model, the price linearly decreases as the total amount of the goods increases.

This oligopoly game is mathematically described as follows.
The action space of player (firm) $j$ is given by $A_j:=[0,x_\mathrm{max}]$ with $0<x_\mathrm{max}<\infty$ for all $j$.
The action of player $j$ is described as $x_j\in A_j$, which represents production of the goods by player $j$.
We collectively write the action profile as $\bm{x}:=\left( x_1, \cdots, x_N \right)$.
We also use the notation $\mathcal{A}:=\prod_j A_j$, $A_{-j}:=\prod_{k\neq j}A_k$, and $x_{-j}:= \left( x_1, \cdots, x_{j-1}, x_{j+1}, \cdots, x_N \right) \in A_{-j}$.
When we want to emphasize $x_j$ in $\bm{x}$, we write $\bm{x}=(x_j, x_{-j})$.
The payoff of player $j$ when the action profile is $\bm{x}$ corresponds to the profit  (the difference between sales and the total production cost) described as
\begin{align}
 s_j \left( \bm{x} \right) &= \left[ a-b\sum_{k=1}^N x_k \right] x_j \theta\left( a-b\sum_{k=1}^N x_k \right) - cx_j,
 \label{eq:payoff_cournot}
\end{align}
where $a>c>0$, $b>0$, and $\theta(\cdot)$ is the step function.
That is, $\left[ a-b\sum_{k=1}^N x_k \right] \theta\left( a-b\sum_{k=1}^N x_k \right)$ represents a price of the goods, and $c$ represents a cost of production.
The price is a non-negative decreasing function of the total production $\sum_{k=1}^N x_k$, and becomes zero when the total production is too large: $\sum_{k=1}^N x_k\geq a/b$.
In this paper, we assume that $x_\mathrm{max}\geq a/b$.

It is known that the Nash equilibrium of this game is $x_j=x^{(\mathrm{N})}:=\frac{a-c}{(N+1)b}$ and $s_j\left( \bm{x}^{(\mathrm{N})} \right)=\frac{(a-c)^2}{(N+1)^2b}$ for all $j$, where $\bm{x}^{(\mathrm{N})}:=\left( x^{(\mathrm{N})}, \cdots, x^{(\mathrm{N})} \right)$.
On the other hand, if all firms collude to monopolize the market, and share the payoffs equally, the realized state is $x_j=x^{(\mathrm{C})}:=\frac{a-c}{2Nb}$ and $s_j\left( \bm{x}^{(\mathrm{C})} \right)=\frac{(a-c)^2}{4Nb}$ for all $j$, where $\bm{x}^{(\mathrm{C})}:=\left( x^{(\mathrm{C})}, \cdots, x^{(\mathrm{C})} \right)$.
Because $s_j\left( \bm{x}^{(\mathrm{C})} \right)>s_j\left( \bm{x}^{(\mathrm{N})} \right)$, the game can be regarded as one example of social dilemmas.
In addition, the Walras equilibrium, which corresponds to perfect competition, is $x_j=x^{(\mathrm{W})}:=\frac{a-c}{Nb}$ and $s_j\left( \bm{x}^{(\mathrm{W})} \right)=0$ for all $j$, where $\bm{x}^{(\mathrm{W})}:=\left( x^{(\mathrm{W})}, \cdots, x^{(\mathrm{W})} \right)$.
In perfect competition, firms cannot influence the price, and the price $\left[ a-b\sum_{k=1}^N x_k \right] \theta\left( a-b\sum_{k=1}^N x_k \right)$ is equal to the marginal cost $c$.
In the limit $N\rightarrow \infty$, the total production in the Nash equilibrium converges to $Nx^{(\mathrm{N})}\rightarrow (a-c)/b$, where $(a-c)/b=Nx^{(\mathrm{W})}$ is the total production in perfect competition.
This result implies that perfect competition is realized if the number of firms is very large, and is known as the Cournot limit theorem.

We consider a repeated version of the Cournot oligopoly game.
We write the probability measure of an action profile at time $t$ by $P_t \left( \cdot \right)$.
The payoff of player $j$ in the repeated game is given by
\begin{align}
 \mathcal{S}_j &:= \left\langle s_j \right\rangle^*,
\end{align}
where $\left\langle \cdots \right\rangle^*$ is the expected value of the quantity $\cdots$ with respect to the limit probability of an action profile
\begin{align}
 P^*\left( \cdot \right) &:= (1-\delta) \sum_{t=1}^\infty \delta^{t-1} P_t \left( \cdot \right)
\end{align}
and $\delta$ is the discount factor satisfying $0\leq \delta <1$.
If the limit $\lim_{T\rightarrow \infty} \sum_{t=1}^T P_t \left( \cdot \right)/T$ exists, we obtain
\begin{align}
 \lim_{\delta\rightarrow 1^-} P^*\left( \cdot \right) &= \lim_{T\rightarrow \infty} \frac{1}{T} \sum_{t=1}^T P_t \left( \cdot \right).
 \label{eq:Pstar_1}
\end{align}
In this paper, we focus only on the undiscounted case $\delta\rightarrow 1^-$.

\section{Preliminaries}
\label{sec:preliminaries}
The memory-one strategy of player $j$ is described as the conditional probability measure $T_j\left( \cdot | \bm{x}^\prime \right)$ of action $x_j$ when the action profile in the previous round is $\bm{x}^{\prime}$.
We introduce the notation $s_0\left( \bm{x} \right):=1$.
For the case of bounded payoffs $\left\{ s_k \right\}_{k=1}^N$, McAvoy and Hauert introduced the concept of autocratic strategies as an extension of zero-determinant (ZD) strategies.
\begin{definition}[\cite{McAHau2016}]
A memory-one strategy of player $j$ is an \emph{autocratic strategy} when its strategy $T_j$ can be written in the form
\begin{align}
 \int \psi(x_j) dT_j\left( x_j | \bm{x}^\prime \right) - \psi\left( x^\prime_j \right) &= \sum_{k=0}^N \alpha_{k} s_{k} \left( \bm{x}^{\prime} \right) 
 \label{eq:AS}
\end{align}
with some coefficients $\left\{ \alpha_{k} \right\}$ and some bounded function $\psi(\cdot)$.
\end{definition}
Surprisingly, the autocratic strategy (\ref{eq:AS}) unilaterally enforces a linear relation between payoffs
\begin{align}
 0 &= \sum_{k=0}^N \alpha_{k} \left\langle s_k \right\rangle^*.
\end{align}
Therefore, autocratic strategies can be used to control payoffs unilaterally.
The left-hand side of Eq. (\ref{eq:AS}) is the difference between an integral of the function $\psi$ with respect to a memory-one strategy $T_j$ and an integral of the function $\psi$ with respect to the Repeat strategy \cite{Aki2016}.
The general meaning of the function $\psi$ is not clear, and it has been heuristically obtained \cite{Ued2022b,Ued2025}.
For two-player potential games, it is related to potential functions \cite{Ued2022}.
Recently, for stochastic games, the method for finding $\psi$ numerically was proposed \cite{MMHet2025}.
It is also noteworthy that autocratic strategies are equally powerful in multi-player games although they were originally introduced in two-player games \cite{McAHau2016}, because their proof did not use the number of players.
In this paper, we use the two words ``ZD strategy'' and ``autocratic strategy'' interchangeably.

Below we write $B\left( \bm{x}^{\prime} \right):=\sum_{k=0}^N \alpha_{k} s_{k} \left( \bm{x}^{\prime} \right)$.
We write the Dirac measure by $\delta_{x^\prime}\left( \cdot \right)$.
McAvoy and Hauert proved the following proposition.
\begin{proposition}[\cite{McAHau2016}]
\label{prop:two-point}
Suppose that there exist two actions $\underline{x}_j, \overline{x}_j \in A_j$ and $W>0$ such that
\begin{align}
 -W &\leq B\left( \underline{x}_j, x_{-j} \right) \leq 0 \quad (\forall x_{-j} \in A_{-j}) \nonumber \\
 0 &\leq B\left( \overline{x}_j, x_{-j} \right) \leq W \quad (\forall x_{-j} \in A_{-j}).
 \label{eq:existence_condition}
\end{align}
Then, when we restrict the action set of player $j$ from $A_j$ to $A_j^\prime:=\left\{ \underline{x}_j, \overline{x}_j \right\}$, the memory-one strategy of player $j$
\begin{align}
 T_j\left( \cdot | \bm{x}^\prime \right) &= p\left( \bm{x}^\prime \right) \delta_{\underline{x}_j} (\cdot) + \left( 1-p\left( \bm{x}^\prime \right) \right) \delta_{\overline{x}_j} (\cdot)
 \label{eq:transition_two-point}
\end{align}
with
\begin{align}
 p\left( \underline{x}_j, x_{-j} \right) &:= \frac{1}{W} B\left( \underline{x}_j, x_{-j} \right) + 1 \quad (\forall x_{-j} \in A_{-j}) \nonumber \\
 p\left( \overline{x}_j, x_{-j} \right) &:= \frac{1}{W} B\left( \overline{x}_j, x_{-j} \right) \quad (\forall x_{-j} \in A_{-j})
 \label{eq:p_two-point}
\end{align}
is an autocratic strategy unilaterally enforcing $\left\langle B \right\rangle^*=0$.
\end{proposition}
Indeed, we can find
\begin{align}
 \int_{x_j \in A_j^\prime} \psi(x_j) dT_j\left( x_j | \bm{x}^\prime \right) - \psi\left( x^\prime_j \right) &= B \left( \bm{x}^{\prime} \right) \quad \left( \forall x^\prime_j \in A_j^\prime, \forall x^\prime_{-j} \in A_{-j} \right)
\end{align}
with $\psi\left( \underline{x}_j \right)=\psi_0 + W$ and $\psi\left( \overline{x}_j \right)=\psi_0$ $(\psi_0\in \mathbb{R})$.
McAvoy and Hauert called such autocratic strategies \emph{two-point autocratic strategies}, because such an autocratic strategy uses only two actions.
This proposition on two-point autocratic strategies is useful, since we can easily construct an autocratic strategy only by specifying the two actions $\underline{x}_j$ and $\overline{x}_j$.

Although they found that the condition (\ref{eq:existence_condition}) is a sufficient condition for the existence of general autocratic strategies, Ueda proved that it is also a necessary condition when the number of actions of each player is finite \cite{Ued2022b}.
However, we do not know whether it is also a necessary condition when the number of actions of some player is infinite.
At this stage, we only prove the following proposition about two-point autocratic strategies.
\begin{proposition}
\label{prop:existence_cournot}
If a two-point autocratic strategy (which uses only two actions) of player $j$ controlling $B$ exists, then there exist two actions $\underline{x}_j, \overline{x}_j \in A_j$ such that
\begin{align}
 B\left( \underline{x}_j, x_{-j} \right) &\leq 0 \quad (\forall x_{-j} \in A_{-j}) \nonumber \\
 B\left( \overline{x}_j, x_{-j} \right) &\geq 0 \quad (\forall x_{-j} \in A_{-j}).
 \label{eq:existence_condition_ueda}
\end{align}
\end{proposition}
See Section \ref{subsec:existence_twopoint} for the proof.
Because $B$ is bounded, the condition in Proposition \ref{prop:existence_cournot} is equivalent to that in Proposition \ref{prop:two-point}.
Therefore, specifying two actions $\underline{x}_j$ and $\overline{x}_j$ satisfying the condition (\ref{eq:existence_condition}) is equivalent to specifying a two-point autocratic strategy.
These propositions are frequently used in order to specify necessary and sufficient conditions for the existence of several types of two-point autocratic strategies.

\section{Theoretical results}
\label{sec:results}
In this section, we investigate the existence of autocratic strategies in the Cournot oligopoly game.
We remark that, in the Cournot oligopoly game, the relation
\begin{align}
  \sum_{k\neq j} s_{k} \left( \bm{x}^{\prime} \right) &= \left[ a-b\sum_{l=1}^N x_l \right] \sum_{k\neq j}x_k \theta\left( a-b\sum_{l=1}^N x_l \right) - c \sum_{k\neq j}x_k
  \label{eq:smj}
\end{align}
holds.
Therefore, from the viewpoint of player $j$, the opponents $-j$ can be essentially regarded as one player with the action $\sum_{k\neq j}x_k$ and the payoff $\sum_{k\neq j} s_{k}$.
However, the domain of the action of ``player'' $-j$ is different from that of player $j$, since $\sum_{k\neq j}x_k \in [0, (N-1)x_\mathrm{max}]$.
Taking this difference into account, as far as we focus on the relation between $s_j$ and $\sum_{k\neq j} s_{k}$, we identify $x_{-j}$ with $\sum_{k\neq j}x_k$ below.

\subsection{Equalizer strategy}
First, we seek for equalizer strategies of player $j$ \cite{BNS1997,PreDys2012}, that is, those unilaterally enforcing $\sum_{k\neq j} \left\langle s_k \right\rangle^*=r$ with $r\in \mathbb{R}$.
For such a case, we need to set $B$ as $B\left( \bm{x}^{\prime} \right):=\sum_{k\neq j} s_{k} \left( \bm{x}^{\prime} \right) - r$.
\begin{theorem}
\label{thm:equalizer}
Two-point equalizer strategies in the Cournot oligopoly game do not exist for any $r$.
\end{theorem}
See Section \ref{subsec:proof_equalizer} for the proof.
Therefore, players cannot unilaterally set the opponents' total payoffs by two-point autocratic strategies.
This result is quite different from that in the repeated prisoner's dilemma game, because two-point equalizer strategies exist in that game.
(It should be noted that all ZD strategies in two-action games are two-point autocratic strategies.)

\subsection{Strategy pinning its own payoff}
Next, we seek for autocratic strategies of player $j$ pinning its own payoff, that is, those unilaterally enforcing $\left\langle s_j \right\rangle^*=r$ with $r\in \mathbb{R}$.
For such a case, we need to set $B$ as $B\left( \bm{x}^{\prime} \right):= s_{j} \left( \bm{x}^{\prime} \right) - r$.
\begin{theorem}
\label{thm:self_pinning}
Two-point autocratic strategies pinning its own payoff in the Cournot oligopoly game exist only for $-cx_\mathrm{max}\leq r \leq 0$.
\end{theorem}
See Section \ref{subsec:proof_self_pinning} for the proof.
We find that $\overline{x}_j=0$ and $\underline{x}_j=x_\mathrm{max}$ if such two-point autocratic strategies exist.

This result is also quite different from that in the repeated prisoner's dilemma game, because ZD strategies pinning its own payoff do not exist in that game \cite{PreDys2012}.
We also remark that the result for $r=0$ seems to be intuitive, because it is realized by producing nothing $(x_j=0)$ in all rounds.
Theorem \ref{thm:self_pinning} claims that such control of its own payoff is possible even for $-cx_\mathrm{max}\leq r < 0$.
This comes from the fact that the one-shot payoff (\ref{eq:payoff_cournot}) can be unilaterally controlled to a negative value if production of the player is too large: $s_j \left( \bm{x} \right) = - cx_j$ for $x_j\geq a/b$.
However, players have no incentive to adopt such autocratic strategies, because players may obtain positive payoffs by using other strategies.

\subsection{Positively correlated strategies}
Here, we seek for positively correlated autocratic strategies of player $j$ \cite{CheZin2014,MMI2022}, which unilaterally enforce
\begin{align}
 \left\langle s_j \right\rangle^* - \kappa &= \chi \left( \frac{1}{N-1} \sum_{k\neq j} \left\langle s_k \right\rangle^* - \kappa \right)
 \label{eq:linear_extortionate}
\end{align}
with $\chi > 0$ and $-cx_\mathrm{max}\leq \kappa \leq (a-c)^2/(4b)$.
For such a case, we need to set $B$ as
\begin{align}
 B\left( \bm{x}^{\prime} \right) &:= s_{j}\left( \bm{x}^{\prime} \right) - \frac{\chi}{N-1} \sum_{k\neq j} s_{k} \left( \bm{x}^{\prime} \right) + (\chi-1)\kappa.
 \label{eq:B_extortionate}
\end{align}
\begin{theorem}
\label{thm:extortionate}
Two-point positively correlated autocratic strategies exist only for (i) $\chi= 1$ or (ii) $0<\chi < 1$ and $\kappa\leq 0$.
\end{theorem}
See Section \ref{subsec:proof_extortionate} for the proof.
We find that $\overline{x}_j=\chi/(\chi+N-1) \cdot (a-c)/b$ and $\underline{x}_j=x_\mathrm{max}$ if such autocratic strategies exist.

This result is also different from that in the repeated prisoner's dilemma game, because two-point positively correlated ZD strategies exist only for $\chi\geq 1$ in that game \cite{PreDys2012,StePlo2013}.
For the case $0<\chi<1$, if $\left\langle s_j \right\rangle^* - \kappa \gtrless 0$, Eq. (\ref{eq:linear_extortionate}) implies $\left\langle s_j \right\rangle^* - \kappa \lessgtr \sum_{k\neq j} \left\langle s_k \right\rangle^*/(N-1) - \kappa$.
Therefore, setting $0<\chi<1$ and $\kappa=0$ leads to low-risk low-return autocratic strategies.
We also remark that, when $\chi \rightarrow 0$, the two-point positively correlated autocratic strategy is reduced to the autocratic strategy in Theorem \ref{thm:self_pinning}, since $\overline{x}_j=\chi/(\chi+N-1) \cdot (a-c)/b \rightarrow 0$.
Moreover, when $\chi = 1$, the two-point positively correlated autocratic strategy can be regarded as a \emph{fair autocratic strategy} \cite{HWTN2014}, since it unilaterally enforces
\begin{align}
 \left\langle s_j \right\rangle^* &= \frac{1}{N-1} \sum_{k\neq j} \left\langle s_k \right\rangle^*,
 \label{eq:linear_aus}
\end{align}
that is, it is guaranteed to obtain the average payoff of the opponents.
When $N=2$, it is reduced to the autocratic strategy which unilaterally equalizes the payoffs of two players \cite{Ued2022,Ued2022b}.
Such a strategy is contained in the class of unbeatable strategies \cite{DOS2012b,DOS2014} (or, rival strategies \cite{HTS2015}), which always obtain payoffs no less than that of the opponent irrespective of the opponent's strategy in two-player symmetric games.
For $N\geq 3$, the two-point positively correlated autocratic strategy with $\chi=1$ can be interpreted as an unbeatable strategy with respect to group average.
Note that, when $\chi =1$, the action $\overline{x}_j=\chi/(\chi+N-1) \cdot (a-c)/b=(a-c)/(Nb)$ is one in the Walras equilibrium $x^{(\mathrm{W})}$.
It has been known that the Walras equilibrium action $x^{(\mathrm{W})}$ is an unbeatable action against any monomorphic opponents \cite{Veg1997}.
The fair autocratic strategy uses this property of the Walras equilibrium action to control payoffs in the repeated game.

\subsection{Negatively correlated strategies}
Here, we seek for negatively correlated autocratic strategies of player $j$, which unilaterally enforce
\begin{align}
 \left\langle s_j \right\rangle^* - \kappa &= \chi \left( \frac{1}{N-1} \sum_{k\neq j} \left\langle s_k \right\rangle^* - \kappa \right)
 \label{eq:linear_negative}
\end{align}
with $\chi < 0$ and $-cx_\mathrm{max}\leq \kappa \leq (a-c)^2/(4b)$.
For such a case, we again need to set $B$ as
\begin{align}
 B\left( \bm{x}^{\prime} \right) &:= s_{j}\left( \bm{x}^{\prime} \right) - \frac{\chi}{N-1} \sum_{k\neq j} s_{k} \left( \bm{x}^{\prime} \right) + (\chi-1)\kappa.
 \label{eq:B_negative}
\end{align}
\begin{theorem}
\label{thm:negative}
Two-point negatively correlated autocratic strategies exist only for $\left| \chi \right| \leq 1$ and $-cx_\mathrm{max}/\left( 1 + \left| \chi \right| \right) \leq \kappa \leq - \left| \chi \right| cx_\mathrm{max}/\left( 1 + \left| \chi \right| \right)$.
\end{theorem}
See Section \ref{subsec:proof_negative} for the proof.
We find that $\overline{x}_j=0$ and $\underline{x}_j=x_\mathrm{max}$ if such two-point autocratic strategies exist.

It should be noted that when $\chi \rightarrow -0$, the two-point negatively correlated autocratic strategy is also reduced to the autocratic strategy in Theorem \ref{thm:self_pinning}.
When $\chi=-1$, $\kappa=-cx_\mathrm{max}/2$ must hold, and the two-point negatively correlated autocratic strategy is similar to the ZD strategy pinning the sum of payoffs of two players in the prisoner's dilemma game \cite{HTS2015,Ued2021b}.
However, for the Cournot oligopoly game, $\kappa<0$ must hold, and the use of such autocratic strategies may be limited.

\section{Numerical results}
\label{sec:numerical}
In this section, we numerically investigate performance of the autocratic strategies in the previous section.
We set model parameters $a=2.0$, $b=1.0$, $c=1.0$, and $x_\mathrm{max}=2.5$.
Because $B$ in Eqs. (\ref{eq:B_extortionate}) and (\ref{eq:B_negative}) for the two-point autocratic strategies satisfies
\begin{align}
 \left| B\left( \bm{x}^{\prime} \right) \right| &\leq \left| s_{j}\left( \bm{x}^{\prime} \right) \right| + \frac{\left| \chi \right|}{N-1} \sum_{k\neq j} \left| s_{k} \left( \bm{x}^{\prime} \right) \right| + \left| (\chi-1)\kappa \right| \nonumber \\
 &\leq (1+\left|\chi \right|) \max\left\{ \frac{(a-c)^2}{4b},  cx_\mathrm{max} \right\} + \left| (\chi-1)\kappa \right|,
\end{align}
we set
\begin{align}
 W &= (1+\left|\chi \right|) \max\left\{ \frac{(a-c)^2}{4b},  cx_\mathrm{max} \right\} + \left| (\chi-1)\kappa \right|.
\end{align}
We approximately calculate the payoffs $\mathcal{S}_k$ in the repeated game by using time average of one sample: $\sum_{t=1}^{T} s_k^{(t)}/T$.
We remark that the limit (\ref{eq:Pstar_1}) exists in all simulations below because we use fixed finite-memory strategies or adaptive memory-one strategies which seem to converge, as the opponents' strategies.

\subsection{Autocratic strategies against fixed memory-zero opponents}
First, we consider situations where player $j$ uses the two-point autocratic strategies with various $(\chi, \kappa)$ and all opponents repeat fixed actions.
We set $N=10$ and $T=10^6$.
We assume that all other players $k\neq j$ use the same memory-zero strategy
\begin{align}
 T_k\left( \cdot | \bm{x}^\prime \right) &= \delta_{x_k} (\cdot)
\end{align}
with $x_k=0.01\times i_\mathrm{oppo}/(N-1)\times (a-c)/b$, and change the integer $i_\mathrm{oppo}$ in the range $i_\mathrm{oppo}\in [0, 200]$.
The procedure is summarized in Algorithm \ref{algo:extortionate}.

\begin{algorithm}
\caption{Autocratic strategies against fixed memory-zero opponents}
\label{algo:extortionate}
\begin{algorithmic}[1]
\REQUIRE Parameters of models $a, b, c, x_\mathrm{max}, N$, total time $T$
\REQUIRE Feasible parameters of two-point autocratic strategies $\chi, \kappa, W$
\ENSURE Time-averaged payoffs $S_l$ for all $l$
\STATE Set $\overline{x}_j$ and $\underline{x}_j$
\FOR{$i_\mathrm{oppo}=0$ to $200$}
\STATE Initial condition: $x_j=\overline{x}_j$
\FOR{$k\neq j$}
\STATE $x_k=0.01\times i_\mathrm{oppo}/(N-1)\times (a-c)/b$
\ENDFOR
\STATE Initialize time-averaged payoffs $S_l \gets 0$ for all $l$
\STATE Calculate $B$ from Eq. (\ref{eq:B_extortionate}) or (\ref{eq:B_negative})
\FOR{$t=1$ to $T$}
\STATE Update $x_j$ by transition probability (\ref{eq:transition_two-point}) and (\ref{eq:p_two-point})
\STATE Calculate payoffs (\ref{eq:payoff_cournot}) for all players
\STATE Update total payoffs $S_l \leftarrow S_l + s_l$ for all $l$
\STATE Calculate $B$ from Eq. (\ref{eq:B_extortionate}) or (\ref{eq:B_negative})
\ENDFOR
\STATE Output time-averaged payoffs $S_l/T$ for all $l$
\ENDFOR
\end{algorithmic}
\end{algorithm}

In Fig \ref{fig:linear_extortionate_vs_zero}, we display relations between $\mathcal{S}_j$ and $\sum_{k\neq j} \mathcal{S}_k/(N-1)$ for various $(\chi, \kappa)$.
\begin{figure}[tbp]
\includegraphics[clip, width=8.0cm]{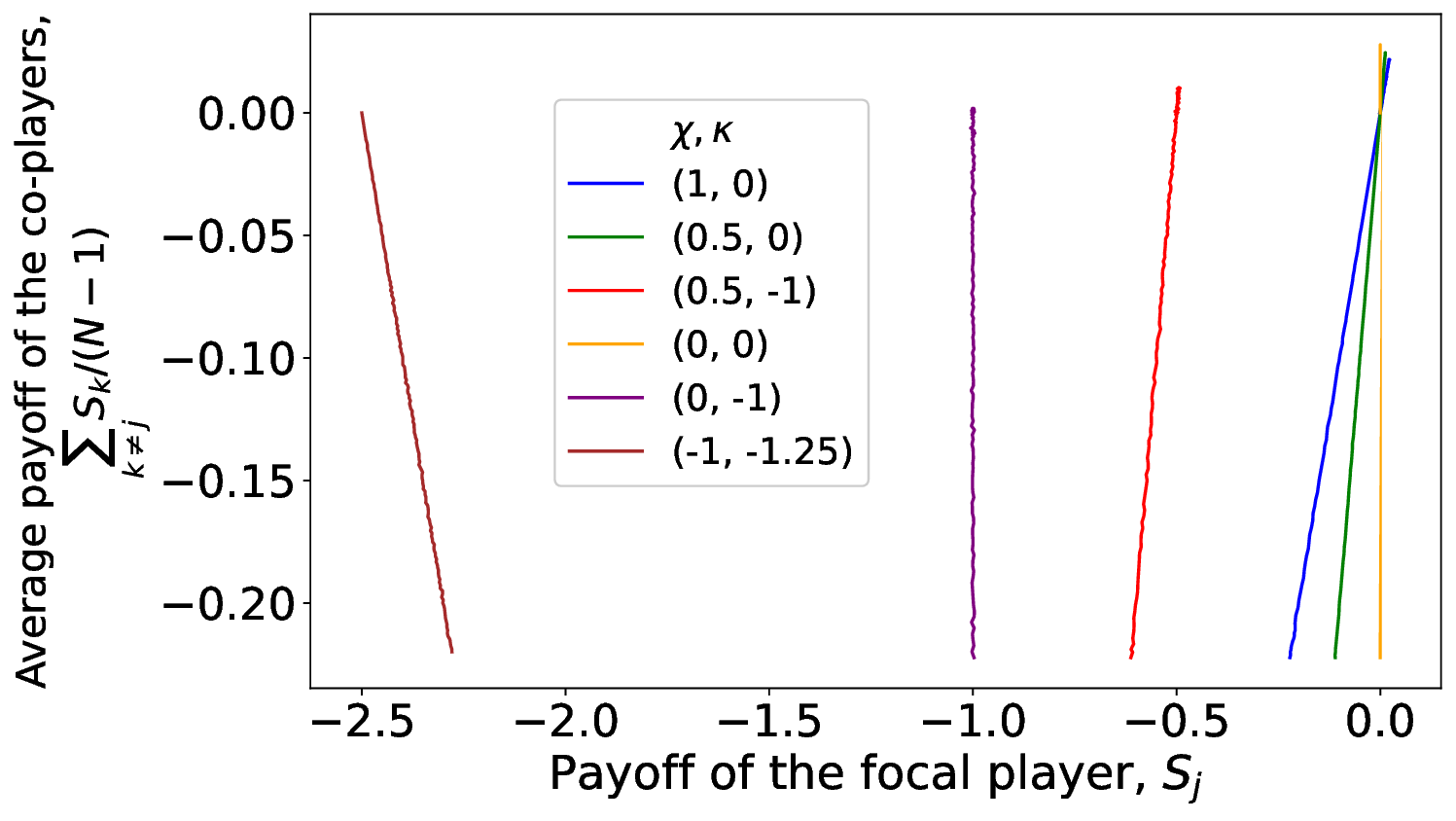}
\caption{Linear relations between $\mathcal{S}_j$ and $\sum_{k\neq j} \mathcal{S}_k/(N-1)$ when player $j$ uses the two-point autocratic strategies with various $(\chi, \kappa)$ and all opponents repeat fixed actions.}
\label{fig:linear_extortionate_vs_zero}
\end{figure}
We find that a linear relation (\ref{eq:linear_extortionate}) or (\ref{eq:linear_negative}) is indeed enforced for each $(\chi, \kappa)$.
We also find that both payoffs $\mathcal{S}_j$ and $\sum_{k\neq j} \mathcal{S}_k/(N-1)$ become positive only for restricted regions, such as the right ends of the lines for $(\chi, \kappa)=(1, 0), (0.5, 0)$.
Therefore, it implies that the use of the autocratic strategies is very limited, considering incentive for adopting such strategies.
Below we focus on the fair autocratic strategy $(\chi, \kappa)=(1, 0)$.

\subsection{Autocratic strategies against random memory-one opponents}
\label{subsec:random_m1}
Next, we consider the situation where player $j$ uses a two-point autocratic strategy with $(\chi, \kappa)=(1, 0)$, and the opponents $k \neq j$ adopt random memory-one strategies.
We assume that each opponent takes either the cooperative action $x^{(\mathrm{C})}$ or the Nash equilibrium action $x^{(\mathrm{N})}$, which are the most representative actions in the Cournot oligopoly game.
A memory-one strategy of player $k$ is written as
\begin{align}
 T_k\left( \cdot | \bm{x}^\prime \right) &= q_k\left( \bm{x}^\prime \right) \delta_{x^{(\mathrm{C})}} (\cdot) + \left( 1- q_k\left( \bm{x}^\prime \right) \right) \delta_{x^{(\mathrm{N})}} (\cdot),
\end{align}
where $q_k\left( \bm{x}^\prime \right)$ corresponds to the probability to use $x^{(\mathrm{C})}$ when the action profile in the previous round was $\bm{x}^\prime$.
For simplicity, we below write $\underline{x}_j$ and $\overline{x}_j$ as $\underline{x}$ and $\overline{x}$, respectively.
We represent the memory-one strategies by a vector \cite{HWTN2014}
\begin{align}
\bm{q}_k &= \left( q_{\underline{x},x^{(\mathrm{C})},0}, \dots, q_{\underline{x},x^{(\mathrm{C})},N-2}, \ q_{\underline{x},x^{(\mathrm{N})},0}, \dots, q_{\underline{x},x^{(\mathrm{N})},N-2}, \right. \nonumber \\
& \qquad \left. q_{\overline{x},x^{(\mathrm{C})},0}, \dots, q_{\overline{x},x^{(\mathrm{C})},N-2}, \ q_{\overline{x},x^{(\mathrm{N})},0}, \dots, q_{\overline{x},x^{(\mathrm{N})},N-2} \right).
\label{eq:memory-one_opponents}
\end{align}
The entries $q_{x_j, x_k, i}$ denote the probability to take $x^{(\mathrm{C})}$ in the next round, given that the autocratic player previously played $x_j \in \{\underline{x}, \overline{x}\}$, player $k$ previously played $x_k \in \{x^{(\mathrm{C})}, x^{(\mathrm{N})}\}$, and $i$ of the other $N-2$ players took $x^{(\mathrm{C})}$. 
We randomly generate $200$ memory-one strategy profiles $\{ \bm{q}_k \}_{k\neq j}$ from a uniform distribution.
The procedure is summarized in Algorithm \ref{algo:random_m1}.
We again set $N=10$ and $T=10^6$.

\begin{algorithm}
\caption{Autocratic strategies against random memory-one opponents}
\label{algo:random_m1}
\begin{algorithmic}[1]
\REQUIRE Parameters of models $a, b, c, x_\mathrm{max}, N$, total time $T$
\REQUIRE Feasible parameters of two-point autocratic strategies $\chi, \kappa, W$
\ENSURE Time-averaged payoffs $S_l$ for all $l$
\STATE Set $\overline{x}_j$ and $\underline{x}_j$
\STATE Set $x^{(\mathrm{C})}$ and $x^{(\mathrm{N})}$
\FOR{$i_\mathrm{oppo}=1$ to $200$}
\STATE Initial condition: $x_j=\overline{x}_j$
\FOR{$k\neq j$}
\STATE Initialize each normal agent's strategy vector $\bm{q}_k \in [0,1]^{4(N-1)}$ $(k\neq j)$ randomly
\STATE Initial condition: $x_k=x^{(\mathrm{C})}$
\ENDFOR
\STATE Initialize time-averaged payoffs $S_l \gets 0$ for all $l$
\STATE Calculate $B$ from Eq. (\ref{eq:B_extortionate}) or (\ref{eq:B_negative})
\FOR{$t=1$ to $T$}
\STATE Update $x_j$ by transition probability (\ref{eq:transition_two-point}) and (\ref{eq:p_two-point})
\FOR{$k\neq j$}
\STATE Update $x_k$ by transition probability (\ref{eq:memory-one_opponents})
\ENDFOR
\STATE Calculate payoffs (\ref{eq:payoff_cournot}) for all players
\STATE Update total payoffs $S_l \leftarrow S_l + s_l$ for all $l$
\STATE Calculate $B$ from Eq. (\ref{eq:B_extortionate}) or (\ref{eq:B_negative})
\ENDFOR
\STATE Output time-averaged payoffs $S_l/T$ for all $l$
\ENDFOR
\end{algorithmic}
\end{algorithm}

Fig \ref{fig:linear_extortionate_vs_one} shows a relationship between the payoff of the autocratic player (horizontal) and the average payoffs of $N-1$ players (vertical).
We find that the autocratic strategy indeed unilaterally enforces a linear relationship between payoffs even against memory-one strategies.  
As shown in Fig \ref{fig:linear_extortionate_vs_zero} and Fig \ref{fig:linear_extortionate_vs_one}, this holds regardless of memory length (0 or 1) or how the opponents' strategies are defined, whether using fixed values with a certain separation or randomized memory-one rules.  
The two-point autocratic strategy consistently enforces the intended linear relation under all these conditions.
Furthermore, both $\mathcal{S}_j$ and $\sum_{k\neq j} \mathcal{S}_k/(N-1)$ are positive for this case.
This result suggests that a player may have incentive to adopt the fair autocratic strategy as far as the opponents use the two representative actions $x^{(\mathrm{C})}$ and $x^{(\mathrm{N})}$.

\begin{figure}[tbp]
\includegraphics[clip, width=8.0cm]{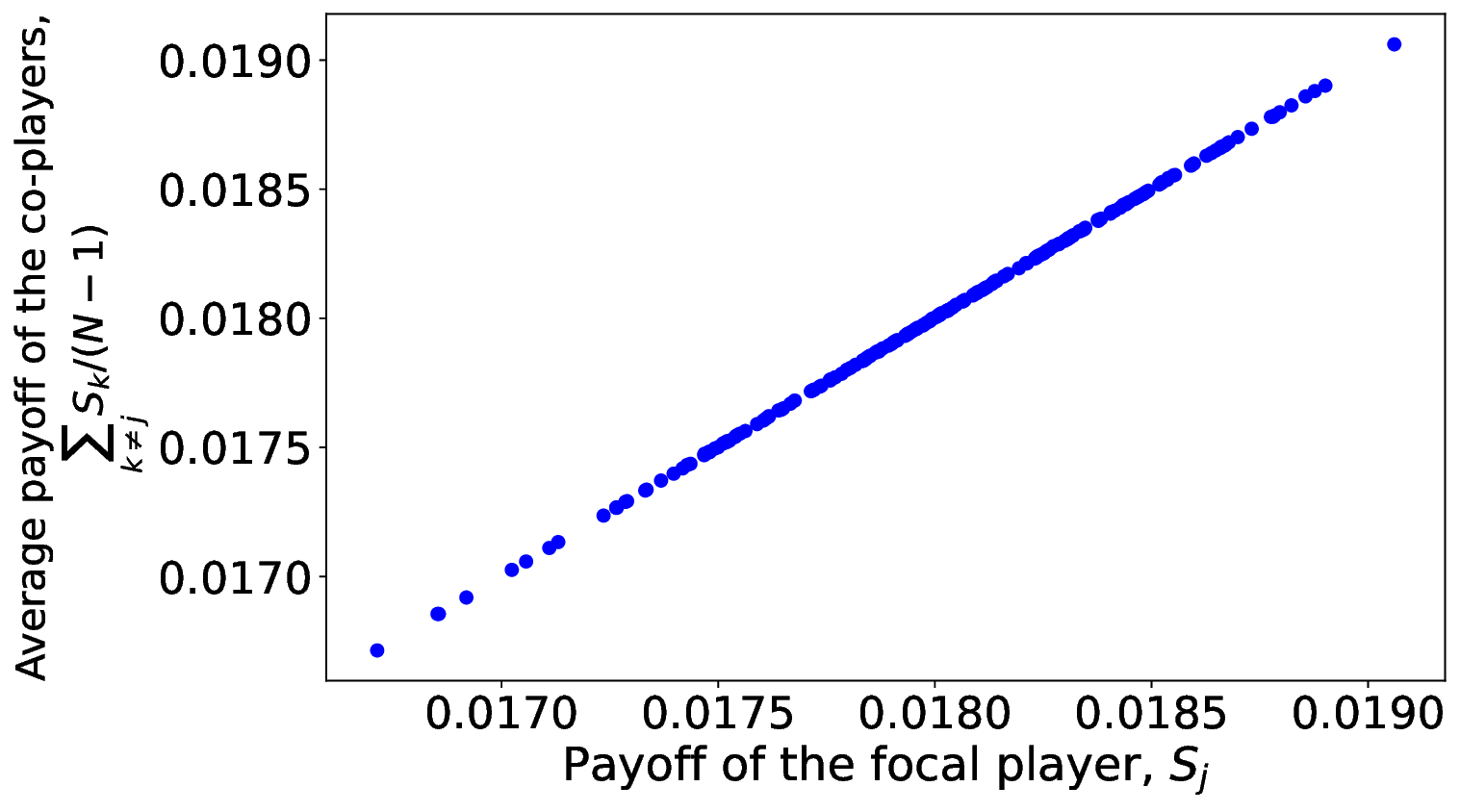}
\caption{A linear relation between $\mathcal{S}_j$ and $\sum_{k\neq j} \mathcal{S}_k/(N-1)$ when player $j$ uses a two-point autocratic strategy with $(\chi, \kappa)=(1, 0)$ against random memory-one strategies.}
\label{fig:linear_extortionate_vs_one}
\end{figure}

\subsection{Autocratic strategies against adaptive memory-one opponents}
Finally, we investigate the scenario in which each memory-one opponent independently attempts to maximize its own payoff against the other players.
We refer to such a player as an adaptive learning player.
In the repeated prisoner's dilemma game, it has been shown that the positively correlated ZD strategy can force an adaptive learning player to cooperate unconditionally~\cite{PreDys2012,CheZin2014,MMI2022}.
Here we again assume that memory-one strategies of the opponents take the form (\ref{eq:memory-one_opponents}), and the autocratic player adopts the fair autocratic strategy $(\chi, \kappa)=(1, 0)$.

To numerically implement this setting, we first initialize actions of all agents and the opponents' memory-one strategies with random values.
The opponents then adaptively update their memory-one strategies using a greedy method independently.
Specifically, at each round, an opponent independently selects one of its strategy parameters $q$ (corresponding to a given previous action profile), and perturbs it by either $+0.01$ or $-0.01$.
Both the original $q$ and the perturbed $\Tilde{q}$ are used to determine the action $x_k$ and $\Tilde{x}_k$, respectively.
Then, the payoff is calculated against the other players.
If the perturbation leads to a higher payoff, the change is accepted; otherwise, it is rejected.
This process is repeated iteratively, allowing each opponent to gradually improve its strategy against the other players.
The detailed algorithm is provided in Algorithm \ref{algo:adaptive_m1}.
We set $N=2$ for Fig \ref{fig:strategy_adaptive_n2} and $N=3$ for Fig \ref{fig:strategy_adaptive_n3}.
Although $T=10^8$ was used, plotting was terminated once no further payoff improvements were observed.

\begin{algorithm}
\caption{Autocratic strategies against adaptive memory-one opponents}
\label{algo:adaptive_m1}
\begin{algorithmic}[1]
\REQUIRE Parameters of models $a, b, c, x_\mathrm{max}, N$, total time $T$
\REQUIRE Feasible parameters of two-point autocratic strategies $\chi, \kappa, W$
\ENSURE Time-averaged payoffs $S_l$ for all $l$ and average strategies $\sum_{k\neq j} \bm{q}_k / (N - 1)$
\STATE Set $\overline{x}_j$ and $\underline{x}_j$
\STATE Set $x^{(\mathrm{C})}$ and $x^{(\mathrm{N})}$
\STATE Initial condition: $x_j \gets$ randomly choose $\overline{x}$ or $\underline{x}$
\FOR{$k\neq j$}
\STATE Initialize each normal agent's strategy vector $\bm{q}_k \in [0,1]^{4(N-1)}$ $(k\neq j)$ randomly
\STATE Initial condition: $x_k \gets$ randomly choose $x^{(\mathrm{C})}$ or $x^{(\mathrm{N})}$
\ENDFOR
\STATE Initialize time-averaged payoffs $S_l \gets 0$ for all $l$
\STATE Calculate $B$ from Eq. (\ref{eq:B_extortionate}) or (\ref{eq:B_negative})
\FOR{$t = 1$ to $T$}
    \FOR{$k \ne j$}
        \STATE Sample action $x_k$ with probability $q_{x_j^\prime, x_{k}^\prime, i^\prime}$ corresponding to previous action profile $\bm{x}^\prime$
        \STATE Create perturbed strategy $\Tilde{q}_{x_j^\prime, x_{k}^\prime, i^\prime} \gets \text{clip}(q_{x_j^\prime, x_{k}^\prime, i^\prime} + \delta, 0, 1)$ with $\delta = \pm 0.01$
        \STATE Sample alternative action $\Tilde{x}_k$ using $\Tilde{q}_{x_j^\prime, x_{k}^\prime, i^\prime}$
        \STATE  Calculate payoffs $s_k$ (from $x_k$) and $\Tilde{s}_k$ (from $\tilde{x}_k$)
        \IF{$\Tilde{s}_k > s_k$}
            \STATE Update $q_{x_j^\prime, x_{k}^\prime, i^\prime} \gets \Tilde{q}_{x_j^\prime, x_{k}^\prime, i^\prime}$, adopt $x_k \gets \Tilde{x}_k$
        \ELSE
            \STATE Keep $q_{x_j^\prime, x_{k}^\prime, i^\prime}$ and $x_k$
        \ENDIF
    \ENDFOR
    \STATE Update $x_j$ by transition probability (\ref{eq:transition_two-point}) and (\ref{eq:p_two-point})
    \STATE Calculate payoffs (\ref{eq:payoff_cournot}) for all players
    \STATE Update total payoffs $S_l \gets S_l + s_l$ for all $l$
    \STATE Calculate $B$ from Eq. (\ref{eq:B_extortionate}) or (\ref{eq:B_negative})
    \STATE Output time-averaged payoffs $S_l / t$ for all $l$ and average strategies $\sum_{k \neq j} \bm{q}_k / (N - 1)$
\ENDFOR
\end{algorithmic}
\end{algorithm}

Fig \ref{fig:strategy_adaptive_n2}a shows the time evolution of the strategy used by the adaptive memory-one player against a fixed autocratic strategy in the case of $N=2$, while Fig \ref{fig:strategy_adaptive_n2}b shows the time evolution of the time-averaged payoffs of both players.
From Fig \ref{fig:strategy_adaptive_n2}a, we observe that the memory-one strategy converges to $q_{\underline{x},x^{(\mathrm{C})},0} = 1$ and $q_{\overline{x},x^{(\mathrm{C})},0} = 1$, indicating that the adaptive memory-one player chooses $x^{(\mathrm{C})}$ with probability one in all subsequent rounds.
In addition, as time progresses, the adaptive player's time-averaged payoff increases and gradually approaches the theoretical expected payoff of $0.05405$ obtained when it consistently chooses $x^{(\mathrm{C})}$ (Fig \ref{fig:strategy_adaptive_n2}b; see Section A in \nameref{S1_Appendix} for derivation).
Furthermore, in order to check the validity of this numerical result, we also performed 100 numerical simulations, and find that 85 runs converge to the similar result as in Fig \ref{fig:strategy_adaptive_n2} (See Section B in \nameref{S1_Appendix}).
In the remaining 15 runs, $q_{\underline{x},x^{(\mathrm{C})},0}$ did not reach 1.
But this is most likely because the total number of time steps ($10^8$) was insufficient. 
In sum, due to the linear payoff relationship between the two players unilaterally enforced by the autocratic strategy, an adaptive memory-one player attempting to maximize its own payoff is inevitably driven to adopt the cooperative action $x^{(\mathrm{C})}$, similarly to the case of the prisoner's dilemma game \cite{PreDys2012,CheZin2014,MMI2022}.

\begin{figure}[tbp]
\includegraphics[clip, width=8.0cm]{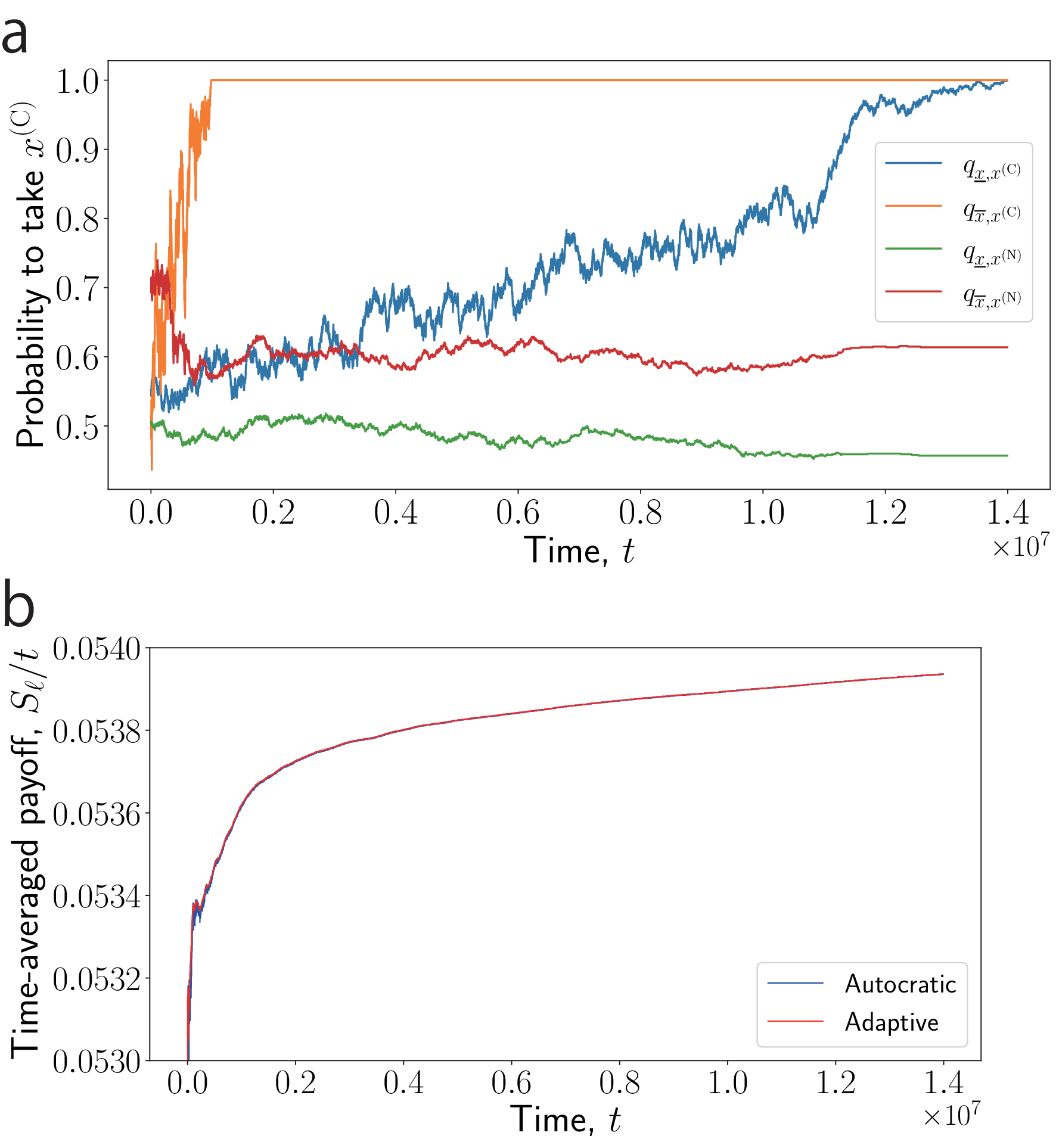}
\caption{Numerical results for in the duopoly case ($N=2$). (a) Strategy adaptation of an adaptive memory-one player against the fair autocratic strategy. (b) Time-averaged payoffs of a fixed autocratic player and an adaptive memory-one player. Data points are plotted every 1000 time steps.}
\label{fig:strategy_adaptive_n2}
\end{figure}

Fig \ref{fig:strategy_adaptive_n3}a shows the time evolution of the strategy used by the adaptive memory-one players against a fixed autocratic strategy in the case of $N=3$, while Fig \ref{fig:strategy_adaptive_n3}b shows the time evolution of the time-averaged payoffs of all players.
According to Fig \ref{fig:strategy_adaptive_n3}b, in the case of $N=3$, the payoffs of the two adaptive memory-one players do not continue to increase but rather converge around 0.04015.
Notably, this value is significantly lower than the theoretical expected payoff of 0.05039, which would be achieved if both adaptive players consistently chose $x^{(\mathrm{C})}$ (see Section A in \nameref{S1_Appendix}).
This result indicates that, in the case of $N=3$, the adaptive players do not evolve their strategies toward cooperation.
Indeed, as shown in Fig \ref{fig:strategy_adaptive_n3}a, the strategy converges to $\sum q_{\overline{x},x^{(\mathrm{C})},1}/2 = 0$, $\sum q_{\overline{x},x^{(\mathrm{N})},0}/2 = 0$, $\sum q_{\underline{x},x^{(\mathrm{C})},1}/2 = 1$, and $\sum q_{\underline{x},x^{(\mathrm{N})},0}/2 = 1$, suggesting that the two adaptive players evolve their strategies in a synchronized manner.
In other words, rather than attempting to maximize their payoffs against the fixed autocratic strategy, the adaptive random memory-one players appear to have evolved strategies that optimize their payoffs against each other.
Again, in order to check the validity of this numerical result, we also performed 100 numerical simulations, and find that all 100 runs converge to the similar result as in Fig \ref{fig:strategy_adaptive_n3} (See Section B in \nameref{S1_Appendix}).
Furthermore, we also tested a Grim/Trigger-like strategy, setting $\bm{q}_k = \left( q_{\underline{x},x^{(\mathrm{C})},0}, q_{\underline{x},x^{(\mathrm{C})},1},  q_{\underline{x},x^{(\mathrm{N})},0}, q_{\underline{x},x^{(\mathrm{N})},1}, q_{\overline{x},x^{(\mathrm{C})},0},  q_{\overline{x},x^{(\mathrm{C})},1}, q_{\overline{x},x^{(\mathrm{N})},0}, q_{\overline{x},x^{(\mathrm{N})},1} \right)=(0,0,0,0,0,1,0,0)$ for the two adaptive players.
We performed ten simulations, and in all cases, the adaptive players' strategies converged to the similar result as in Fig \ref{fig:strategy_adaptive_n3} (See Section B in \nameref{S1_Appendix}).
These results show that the result in Fig \ref{fig:strategy_adaptive_n3} does not depend on the initial conditions and realizations of random variables.

\begin{figure}[tbp]
\includegraphics[clip, width=8.0cm]{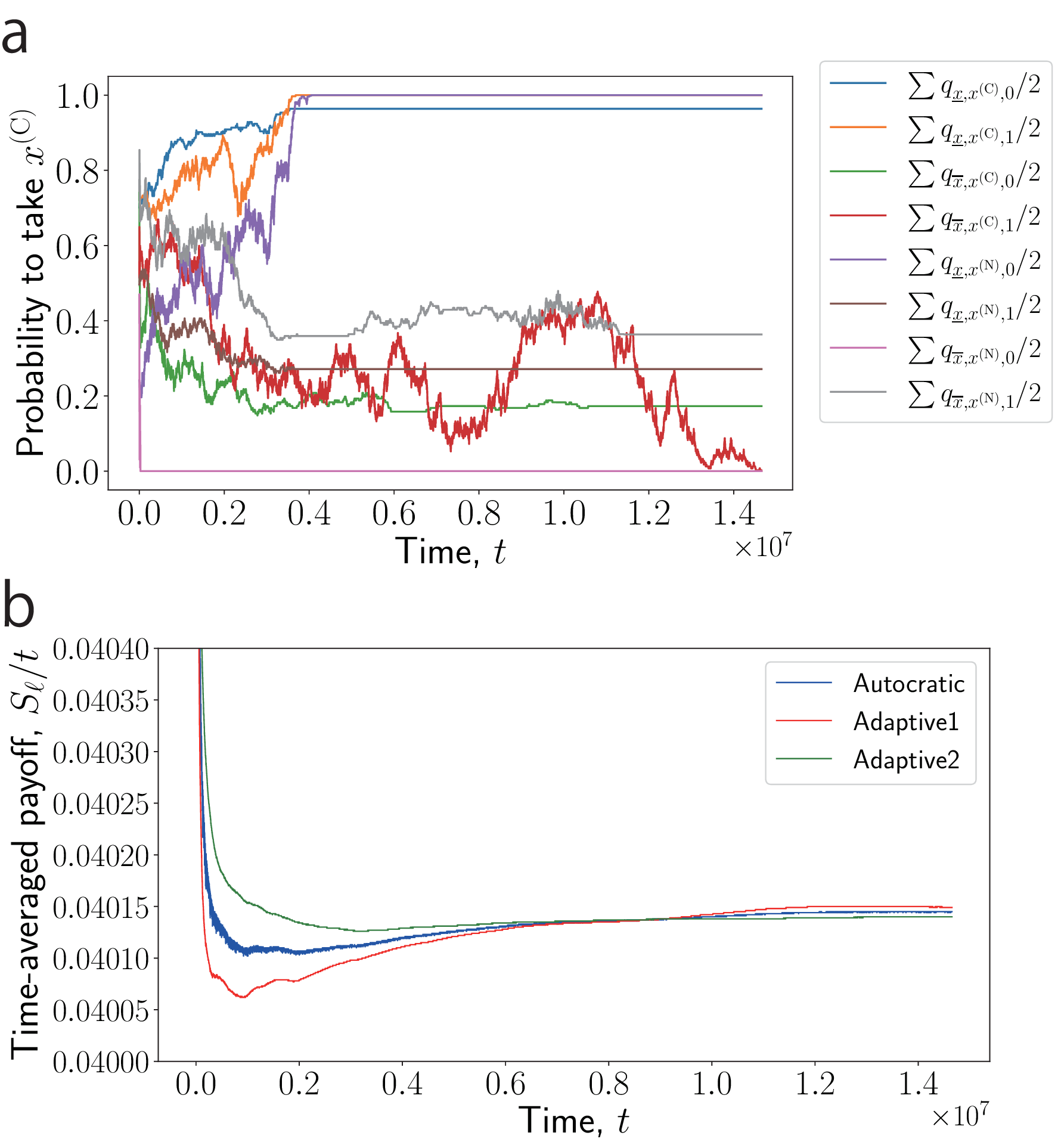}
\caption{Numerical results for in the triopoly case ($N=3$). (a) Strategy adaptation of two adaptive memory-one players against the fair autocratic strategy. (b) Time-averaged payoffs of a fixed autocratic player and two adaptive memory-one players. Data points are plotted every 1000 time steps.}
\label{fig:strategy_adaptive_n3}
\end{figure}

The above results show that the fair autocratic strategy can promote collusion of adaptively learning players in the duopoly ($N=2$) case, but cannot promote collusion in the triopoly ($N=3$) case.
The result in the duopoly case is similar to that in the repeated prisoner's dilemma game \cite{PreDys2012,CheZin2014,MMI2022}.
In the prisoner's dilemma game, the extortionate ZD strategy can force an adaptively learning player (without a theory of mind) to cooperate unconditionally.
This comes from a property of positively correlated ZD strategies: An increase in payoff of the opponent implies an increase in payoff of a player using the ZD strategies.
However, in the triopoly case, the fair autocratic strategy cannot force the adaptive opponents to collude.
This means that enforcing a linear relation (\ref{eq:linear_aus}) is not sufficient to directly control all other opponents; the fair autocratic strategy can control only the average payoff of the opponents.

The failure of control by the fair autocratic strategy in the triopoly case may come from the nature of two adaptive agents.
We provide the values of one-shot payoffs for $N=2$ and $N=3$ in Table \ref{tab:payoff_N2} and Table \ref{tab:payoff_N3}, respectively.
When two adaptive agents regard an autocratic agent as an environment \cite{HSCN2018}, they selfishly learn their best response $x^{(\mathrm{N})}$, since $x^{(\mathrm{N})}$ brings an agent larger one-shot payoff than $x^{(\mathrm{C})}$ as long as an autocratic agent takes $\overline{x}_j$.
(It should be noted that an autocratic agent frequently takes $\overline{x}_j$ and rarely takes $\underline{x}_j$ in our parameters.)
Concretely, in Table \ref{tab:payoff_N3}, when player $1$ is an autocratic agent, the payoff of player $2$ satisfies $s_2\left( \overline{x}, x^{(\mathrm{C})}, x^{(\mathrm{C})} \right)<s_2\left( \overline{x}, x^{(\mathrm{N})}, x^{(\mathrm{C})} \right)$ and $s_2\left( \overline{x}, x^{(\mathrm{C})}, x^{(\mathrm{N})} \right)=s_2\left( \overline{x}, x^{(\mathrm{N})}, x^{(\mathrm{N})} \right)$.
Therefore, they come to repeat $x^{(\mathrm{N})}$ for most rounds.
Only when the state of ``environment'' $x_j$ switches to $\underline{x}_j$, they synchronously take $x^{(\mathrm{C})}$, since $s_2\left( \underline{x}, x^{(\mathrm{C})}, x^{(\mathrm{C})} \right)>s_2\left( \underline{x}, x^{(\mathrm{N})}, x^{(\mathrm{C})} \right)$ and $s_2\left( \underline{x}, x^{(\mathrm{C})}, x^{(\mathrm{N})} \right)>s_2\left( \underline{x}, x^{(\mathrm{N})}, x^{(\mathrm{N})} \right)$ in Table \ref{tab:payoff_N3}.
This is as if they synchronously take $x^{(\mathrm{C})}$ in order to bring back the state of ``environment'' from $\underline{x}_j$ to $\overline{x}_j$, because transition between states of ``environment'' occurs with probability proportional to $\left| B(\bm{x}^\prime) \right|$.
In Section C in \nameref{S1_Appendix}, we provide a numerical result in a situation where only two adaptive agents exist.
Both adaptive agents finally learned to repeat the Nash equilibrium action $x^{(\mathrm{N})}$.
The failure of control in the triopoly case seems to come from the same reason as this situation.
It should be noted that a similar phenomenon occurs in the repeated public goods game \cite{LiHao2019}, where a cooperation-enforcing strategy is adopted against two independent agents using reinforcement learning.
The cooperation-enforcing strategy is known to be an unbeatable ZD strategy \cite{Ued2023}.
Robustness of our numerical result against other learning algorithms should be investigated in future.
In addition, we must remark that our numerical simulations are performed for only one set of payoff parameters $(a, b, c, x_\mathrm{max})$.
For other sets of parameters, the fair autocratic strategy may induce collusion even in the triopoly case.

\begin{table}[!ht]
\begin{adjustwidth}{-2.25in}{0in} 
\centering
\caption{
{\bf Payoffs for $N=2$.} Player $1$ is an autocratic agent, and player $2$ takes either $x^{(\mathrm{C})}$ or $x^{(\mathrm{N})}$.}
\begin{tabular}{|c+c|c|}
\hline
& $x^{(\mathrm{C})}$ & $x^{(\mathrm{N})}$ \\ \thickhline
$\overline{x}$ & $\frac{1}{8}\frac{(a-c)^2}{b}, \frac{1}{16}\frac{(a-c)^2}{b}$ & $\frac{1}{12}\frac{(a-c)^2}{b}, \frac{1}{18}\frac{(a-c)^2}{b}$ \\ \hline
$\underline{x}$ & $-cx_\mathrm{max}, -\frac{1}{4}\frac{c(a-c)}{b}$ & $-cx_\mathrm{max}, -\frac{1}{3}\frac{c(a-c)}{b}$  \\ \hline
\end{tabular}
\label{tab:payoff_N2}
\end{adjustwidth}
\end{table}

\begin{table}[!ht]
\begin{adjustwidth}{-2.25in}{0in} 
\centering
\caption{
{\bf Payoffs for $N=3$.} Player $1$ is an autocratic agent, and players $2$ and $3$ take either $x^{(\mathrm{C})}$ or $x^{(\mathrm{N})}$.}
\begin{tabular}{|c+c|c|c|c|}
\hline
& $\left( x^{(\mathrm{C})}, x^{(\mathrm{C})} \right)$ & $\left( x^{(\mathrm{C})}, x^{(\mathrm{N})} \right)$ & $\left( x^{(\mathrm{N})}, x^{(\mathrm{C})} \right)$ & $\left( x^{(\mathrm{N})}, x^{(\mathrm{N})} \right)$ \\ \thickhline
$\overline{x}$ & $\frac{1}{9}\frac{(a-c)^2}{b}, \frac{1}{18}\frac{(a-c)^2}{b}, \frac{1}{18}\frac{(a-c)^2}{b}$ & $\frac{1}{12}\frac{(a-c)^2}{b}, \frac{1}{24}\frac{(a-c)^2}{b}, \frac{1}{16}\frac{(a-c)^2}{b}$ & $\frac{1}{12}\frac{(a-c)^2}{b}, \frac{1}{16}\frac{(a-c)^2}{b}, \frac{1}{24}\frac{(a-c)^2}{b}$ & $\frac{1}{18}\frac{(a-c)^2}{b}, \frac{1}{24}\frac{(a-c)^2}{b}, \frac{1}{24}\frac{(a-c)^2}{b}$ \\ \hline
$\underline{x}$ & $-cx_\mathrm{max}, -\frac{1}{6}\frac{c(a-c)}{b}, -\frac{1}{6}\frac{c(a-c)}{b}$ & $-cx_\mathrm{max}, -\frac{1}{6}\frac{c(a-c)}{b}, -\frac{1}{4}\frac{c(a-c)}{b}$ & $-cx_\mathrm{max}, -\frac{1}{4}\frac{c(a-c)}{b}, -\frac{1}{6}\frac{c(a-c)}{b}$ & $-cx_\mathrm{max}, -\frac{1}{4}\frac{c(a-c)}{b}, -\frac{1}{4}\frac{c(a-c)}{b}$ \\ \hline
\end{tabular}
\label{tab:payoff_N3}
\end{adjustwidth}
\end{table}

Although the fair autocratic strategy cannot promote collusion of adaptive agents for $N=3$ (and probably for $N\geq 4$), our result for $N=2$ suggests another way to collusion.
As a special property of the Cournot oligopoly game, players $-j$ can be regarded as one effective player, as shown in Eq. (\ref{eq:smj}).
Therefore, when players $-j$ use a ZD alliance \cite{HWTN2014}, which is a ZD strategy implemented by a group of players, they can lead the opponent $j$ to collusion.
In Section D in \nameref{S1_Appendix}, we prove the existence of a fair ZD alliance in the Cournot oligopoly game.
We expect that the fair ZD alliance by $N-1$ players can promote collusion of the other player.

\section{Discussion}
\label{sec:discussion}
In this paper, we proved the existence of several autocratic strategies in the repeated Cournot oligopoly game, which unilaterally enforce linear relations between the payoff of the autocratic player and the average payoff of the opponents with both positive and negative slopes.
Particularly, we found that a fair autocratic strategy is useful because it can be used for enforcing the payoffs into positive values.
Furthermore, we numerically showed that a fair autocratic strategy can promote collusion in the duopoly case, although it cannot promote collusion in the triopoly case.

Autocratic strategies are useful because they unilaterally control payoffs irrespective of rationality of other players \cite{HCN2018,MMHet2025}.
This is quite different from previous results in classical game theory.
For example, in the folk theorem \cite{Gib1992}, it is assumed that players are rational (in other words, want to maximize their own payoffs), and how the equilibrium strategies are obtained has not been clear.
In contrast, by using positively correlated autocratic strategies, the autocratic player can force the opponents to optimize the payoff of the autocratic player \cite{PreDys2012}.
This result is not limited to the prisoner's dilemma game.
Indeed, in this paper, we found that such enforcement is possible even in a duopoly game.
Even if there are more firms, once the opponents recognize an autocratic party, they may quit optimizing their payoffs independently and start to collude.
Investigation on such multi-firm situations is a subject of future work.

Additionally, we want to emphasize that the results in this paper do not simply mimic those in the prisoner's dilemma game.
As noted in Section \ref{sec:model}, when action sets of all players are restricted to $\{ x^{(\mathrm{C})}, x^{(\mathrm{N})} \}$ and $N=2$, the Cournot oligopoly game is reduced to the prisoner's dilemma game, where $x^{(\mathrm{C})}$ and $x^{(\mathrm{N})}$ correspond to cooperation and defection, respectively (see Section C in \nameref{S1_Appendix}).
However, when action sets are $[0,x_\mathrm{max}]$, resulting autocratic strategies are quite different from those in the prisoner's dilemma game.
In the prisoner's dilemma case, the two-point equalizer strategies exist \cite{BNS1997,PreDys2012}, contrary to our Theorem \ref{thm:equalizer}.
(It should be noted that all ZD strategies in two-action games are two-point autocratic strategies.)
Similarly, although the autocratic strategies pinning its own payoff do not exist in the prisoner's dilemma game \cite{PreDys2012}, they exist in the Cournot oligopoly game (Theorem \ref{thm:self_pinning}).
Furthermore, whereas our Theorem \ref{thm:extortionate} claims that two-point positively correlated autocratic strategies exist only for $\chi\leq 1$, two-point positively correlated autocratic strategies exist for $\chi\geq 1$ for the prisoner's dilemma game \cite{PreDys2012,StePlo2013}.
These differences also exist even if we compare the Cournot oligopoly game with the public goods game \cite{HWTN2014,PHRT2015}, which is an $N$-player version of the prisoner's dilemma game.
These differences seem to come from the nature of the payoff function of the Cournot oligopoly game.

Obviously, these differences in equalizer strategies and positively correlated autocratic strategies may come from the fact that we only consider two-point autocratic strategies.
However, it is known that, when action sets of all players are finite sets, the existence condition of two-point zero-determinant strategies is equivalent to the existence condition of general zero-determinant strategies \cite{Ued2022b,Ued2025}.
We expect that this equivalence also holds when an action set of some player is an infinite set.
Specifying a necessary and sufficient condition for the existence of zero-determinant strategies in the case that action sets of some players are infinite sets is an important future problem.

\section{Methods}
\label{sec:methods}

\subsection{Proof of Proposition \ref{prop:existence_cournot}}
\label{subsec:existence_twopoint}
If a two-point autocratic strategy $T_j$ of player $j$ controlling $B$ exists, then it satisfies
\begin{align}
 \int_{x_j\in A_j^\prime} \psi(x_j) dT_j\left( x_j | \bm{x}^\prime \right) - \psi\left( x^\prime_j \right) &= B \left( \bm{x}^{\prime} \right) \quad \left( \forall x^\prime_j \in A_j^\prime, \forall x^\prime_{-j} \in A_{-j} \right)
 \label{eq:AS_B}
\end{align}
with some set $A_j^\prime=\left\{ x_j^{(1)}, x_j^{(2)} \right\}$ and some bounded function $\psi(\cdot)$.
Because $A_j^\prime$ consists of two actions and $\psi$ is bounded, there exist
\begin{align}
 \psi_{\mathrm{max}} := \max_{x_j\in A_j^\prime} \psi(x_j) \nonumber \\
 \psi_{\mathrm{min}} := \min_{x_j\in A_j^\prime} \psi(x_j) 
\end{align}
and
\begin{align}
 x_{j,\mathrm{max}} := \arg \max_{x_j\in A_j^\prime} \psi(x_j) \nonumber \\
 x_{j,\mathrm{min}} := \arg \min_{x_j\in A_j^\prime} \psi(x_j).
\end{align}
Due to the normalization condition of probability distribution, we can rewrite Eq. (\ref{eq:AS_B}) as
\begin{align}
  B \left( \bm{x}^{\prime} \right) &= \int_{x_j\in A_j^\prime} \left(\psi(x_j) - \psi_{\mathrm{max}} \right) dT_j\left( x_j | \bm{x}^\prime \right) - \left( \psi\left( x^\prime_j \right) - \psi_{\mathrm{max}} \right) \\
  &= \int_{x_j\in A_j^\prime} \left(\psi(x_j) - \psi_{\mathrm{min}} \right) dT_j\left( x_j | \bm{x}^\prime \right) - \left( \psi\left( x^\prime_j \right) - \psi_{\mathrm{min}} \right).
\end{align}
Then we find
\begin{align}
  B \left( x_{j,\mathrm{max}}, x^\prime_{-j} \right) &= \int_{x_j\in A_j^\prime} \left(\psi(x_j) - \psi_{\mathrm{max}} \right) dT_j\left( x_j | x_{j,\mathrm{max}}, x^\prime_{-j} \right) \leq 0 \quad (\forall x^\prime_{-j} \in A_{-j}) 
\end{align}
and
\begin{align}
  B \left( x_{j,\mathrm{min}}, x^\prime_{-j} \right) &= \int_{x_j\in A_j^\prime} \left(\psi(x_j) - \psi_{\mathrm{min}} \right) dT_j\left( x_j | x_{j,\mathrm{min}}, x^\prime_{-j} \right) \geq 0 \quad (\forall x^\prime_{-j} \in A_{-j}) 
\end{align}
because $T_j$ is probability.
Therefore, we can regard $x_{j,\mathrm{max}}$ and $x_{j,\mathrm{min}}$ as $\underline{x}_j$ and $\overline{x}_j$, respectively.

\subsection{Proof of Theorem \ref{thm:equalizer}}
\label{subsec:proof_equalizer}
According to the propositions above, two actions satisfying
\begin{align}
 \sum_{k\neq j} s_{k} \left( \underline{x}_j, x_{-j} \right) - r &\leq 0 \quad (\forall x_{-j} \in A_{-j}) \nonumber \\
 \sum_{k\neq j} s_{k} \left( \overline{x}_j, x_{-j} \right) - r &\geq 0 \quad (\forall x_{-j} \in A_{-j})
\end{align}
are necessary for the existence of two-point equalizer strategies of player $j$.
For $r>0$, when we choose $x_{-j}=0$, we find
\begin{align}
 \sum_{k\neq j} s_{k} \left( x_j, 0 \right) - r &=-r < 0
\end{align}
for any $x_j$.
Therefore, $\overline{x}_j$ does not exist, and the existence condition is never satisfied.
Similarly, for $r<0$, when we choose $x_{-j}=0$, we find
\begin{align}
 \sum_{k\neq j} s_{k} \left( x_j, 0 \right) - r &=-r > 0
\end{align}
for any $x_j$.
Therefore, $\underline{x}_j$ does not exist, and the existence condition is never satisfied.
For $r=0$, when we choose $x_{-j}=a/b$, we find
\begin{align}
 \sum_{k\neq j} s_{k} \left( x_j, \frac{a}{b} \right) &=-\frac{ca}{b} < 0
\end{align}
for any $x_j$.
Therefore, $\overline{x}_j$ does not exist, and the existence condition is never satisfied.

\subsection{Proof of Theorem \ref{thm:self_pinning}}
\label{subsec:proof_self_pinning}
According to the propositions above, two actions satisfying
\begin{align}
 s_{j} \left( \underline{x}_j, x_{-j} \right) - r &\leq 0 \quad (\forall x_{-j} \in A_{-j}) \nonumber \\
 s_{j} \left( \overline{x}_j, x_{-j} \right) - r &\geq 0 \quad (\forall x_{-j} \in A_{-j})
\end{align}
are necessary for the existence of the two-point autocratic strategies of player $j$.
When $r>0$, when we choose $x_{-j}=a/b$, we find
\begin{align}
 s_{k} \left( x_j, \frac{a}{b} \right) - r &= -cx_j - r < 0
\end{align}
for any $x_j$.
Therefore, $\overline{x}_j$ does not exist, and the existence condition is never satisfied.

We remark that $-cx_\mathrm{max}\leq s_{j} \left( \bm{x}^{\prime} \right) \leq (a-c)^2/(4b)$ for any $\bm{x}^{\prime}$.
Therefore, the two-point autocratic strategies with $r<-cx_\mathrm{max}$ do not exist.
When $-cx_\mathrm{max}\leq r \leq 0$, we find that 
\begin{align}
 s_{j} \left( 0, x_{-j} \right) - r &= \left| r \right| \geq 0 \quad (\forall x_{-j} \in A_{-j}) \nonumber \\
 s_{j} \left( x_\mathrm{max}, x_{-j} \right) - r &= -cx_\mathrm{max} + \left| r \right| \leq 0 \quad (\forall x_{-j} \in A_{-j}).
\end{align}
Therefore, according to Proposition \ref{prop:two-point}, we can construct two-point autocratic strategies pinning its own payoff by using $\overline{x}_j=0$ and $\underline{x}_j=x_\mathrm{max}$.

\subsection{Proof of Theorem \ref{thm:extortionate}}
\label{subsec:proof_extortionate}
According to the propositions above, two actions satisfying
\begin{align}
 B \left( \underline{x}_j, x_{-j} \right) &\leq 0 \quad (\forall x_{-j} \in A_{-j}) \nonumber \\
 B \left( \overline{x}_j, x_{-j} \right) &\geq 0 \quad (\forall x_{-j} \in A_{-j})
\end{align}
are necessary for the existence of two-point positively correlated autocratic strategies of player $j$.

We first rewrite the payoffs as
\begin{align}
 s_j \left( \bm{x} \right) &= \left[ \left(a-c-b\sum_{k=1}^N x_k \right) \theta\left( a-b\sum_{k=1}^N x_k \right) - c \theta\left( -\left( a-b\sum_{k=1}^N x_k \right) \right) \right] x_j.
\end{align}
Then, $B$ is rewritten as
\begin{align}
 & B \left( \bm{x} \right) \nonumber \\
 &= \left[ \left(a-c-b\sum_{l=1}^N x_l \right) \theta\left( a-b\sum_{l=1}^N x_l \right) - c \theta\left( -\left( a-b\sum_{l=1}^N x_l \right) \right) \right] \left[ x_j - \frac{\chi}{N-1} \sum_{k\neq j} x_k \right] \nonumber \\
 & \quad + (\chi-1)\kappa.
 \label{eq:B_extortionate_mod}
\end{align}
We find that
\begin{align}
 & B \left( \frac{\chi}{\chi+N-1} \frac{a-c}{b}, x_{-j} \right) \nonumber \\
 & = \frac{\chi}{N-1} b \left( \frac{N-1}{\chi+N-1}\frac{a-c}{b} - \sum_{k\neq j} x_k \right)^2 \theta\left( c + \frac{N-1}{\chi+N-1}(a-c) - b\sum_{k\neq j} x_k \right) \nonumber \\
 & \quad - c \frac{\chi}{N-1} \left( \frac{N-1}{\chi+N-1} \frac{a-c}{b} - \sum_{k\neq j} x_k \right) \theta\left( - \left( c + \frac{N-1}{\chi+N-1}(a-c)  - b\sum_{k\neq j} x_k \right) \right) \nonumber \\
 & \quad + (\chi-1)\kappa.
\end{align}
The first term on the right-hand side is always non-negative.
It should be noted that, when $c + (N-1)(a-c)/(\chi+N-1)- b\sum_{k\neq j} x_k<0$, the second term satisfies
\begin{align}
 - c \frac{\chi}{N-1} \left( \frac{N-1}{\chi+N-1} \frac{a-c}{b} - \sum_{k\neq j} x_k \right) &> c \frac{\chi}{N-1} \frac{c}{b} >0.
\end{align}
Therefore, we obtain
\begin{align}
 B \left( \frac{\chi}{\chi+N-1} \frac{a-c}{b}, x_{-j} \right) &\geq (\chi-1)\kappa.
 \label{eq:B_extortion_over}
\end{align}
Furthermore, 
\begin{align}
 B \left( x_\mathrm{max}, x_{-j} \right) &= -c \left[ x_\mathrm{max} - \frac{\chi}{N-1} \sum_{k\neq j} x_k \right] + (\chi-1)\kappa.
 \label{eq:B_extortion_under}
\end{align}

We also calculate
\begin{align}
 & B \left( x_j, \frac{N-1}{\chi+N-1} \frac{a-c}{b} \right) \nonumber \\
 &= - b \left( x_j - \frac{\chi}{\chi+N-1}\frac{a-c}{b} \right)^2 \theta\left( c + \frac{\chi}{\chi+N-1}(a-c) - bx_j \right) \nonumber \\
 & \quad - c \left( x_j - \frac{\chi}{\chi+N-1} \frac{a-c}{b} \right) \theta\left( - \left( c + \frac{\chi}{\chi+N-1}(a-c)  - bx_j \right) \right) \nonumber \\
 & \quad + (\chi-1)\kappa.
\end{align}
The first term on the right-hand side is always non-positive.
It should be noted that, when $c + \chi(a-c)/(\chi+N-1)- bx_j<0$, the second term satisfies
\begin{align}
  -c \left( x_j - \frac{\chi}{\chi+N-1} \frac{a-c}{b} \right) &< - \frac{c^2}{b} <0.
\end{align}
Therefore, we obtain
\begin{align}
  B \left( x_j, \frac{N-1}{\chi+N-1} \frac{a-c}{b} \right) &\leq (\chi-1)\kappa.
  \label{eq:B_extortion_no}
\end{align}

We first consider the case $0<\chi <1$.
When $\kappa \leq 0$, according to Eq. (\ref{eq:B_extortion_over}), we find
\begin{align}
 B \left( \frac{\chi}{\chi+N-1} \frac{a-c}{b}, x_{-j} \right) &\geq 0.
\end{align}
Furthermore, by using $0\leq \sum_{k\neq j} x_k \leq (N-1) x_\mathrm{max}$ in Eq. (\ref{eq:B_extortion_under}), we obtain
\begin{align}
 B \left( x_\mathrm{max}, x_{-j} \right) &\leq (1-\chi) \left[ \left| \kappa \right| - cx_\mathrm{max} \right] \leq 0,
\end{align}
where we have used $-cx_\mathrm{max}\leq \kappa \leq (a-c)^2/(4b)$.
Therefore, according to Proposition \ref{prop:two-point} with $\overline{x}_j=\chi/(\chi+N-1) \cdot (a-c)/b$ and $\underline{x}_j=x_\mathrm{max}$, two-point positively correlated autocratic strategies exist for $0<\chi <1$ and $\kappa \leq 0$.
When $\kappa>0$, according to Eq. (\ref{eq:B_extortion_no}), we find
\begin{align}
  B \left( x_j, \frac{N-1}{\chi+N-1} \frac{a-c}{b} \right) &\leq (\chi-1)\kappa < 0
\end{align}
for any $x_j$.
Therefore, $\overline{x}_j$ does not exist, and the existence condition is never satisfied.

Second, we consider the case $\chi=1$.
For this case, $B$ does not depend on $\kappa$.
According to Eqs. (\ref{eq:B_extortion_over}) and (\ref{eq:B_extortion_under}), we obtain
\begin{align}
 B \left( \frac{\chi}{\chi+N-1} \frac{a-c}{b}, x_{-j} \right) &\geq 0 \label{eq:xov_extortionate_chi1} \\
 B \left( x_\mathrm{max}, x_{-j} \right) = -c \left[ x_\mathrm{max} - \frac{1}{N-1} \sum_{k\neq j} x_k \right]  &\leq 0.
\end{align}
Therefore, according to Proposition \ref{prop:two-point} with $\overline{x}_j=\chi/(\chi+N-1) \cdot (a-c)/b$ and $\underline{x}_j=x_\mathrm{max}$, two-point positively correlated autocratic strategies exist for $\chi=1$.

Finally, we consider the case $\chi>1$.
When $\kappa \geq 0$, from Eq. (\ref{eq:B_extortionate_mod}), we find
\begin{align}
 B \left( x_j, (N-1) x_\mathrm{max} \right)  &= -c\left[ x_j - \chi x_\mathrm{max}  \right] + (\chi-1)\kappa > 0
\end{align}
for any $x_j$.
Therefore, $\underline{x}_j$ does not exist, and the existence condition is never satisfied.
When $\kappa < 0$, according to Eq. (\ref{eq:B_extortion_no}), we obtain
\begin{align}
  B \left( x_j, \frac{N-1}{\chi+N-1} \frac{a-c}{b} \right) &\leq (\chi-1)\kappa < 0
\end{align}
for any $x_j$.
Therefore, $\overline{x}_j$ does not exist, and the existence condition is never satisfied.

\subsection{Proof of Theorem \ref{thm:negative}}
\label{subsec:proof_negative}
According to the propositions above, two actions satisfying
\begin{align}
 B \left( \underline{x}_j, x_{-j} \right) &\leq 0 \quad (\forall x_{-j} \in A_{-j}) \nonumber \\
 B \left( \overline{x}_j, x_{-j} \right) &\geq 0 \quad (\forall x_{-j} \in A_{-j})
\end{align}
are necessary for the existence of two-point negatively correlated autocratic strategies of player $j$.
We remark that $B$ can be explicitly written as
\begin{align}
 B \left( \bm{x} \right) &= \left[ \left(a-b\sum_{l=1}^N x_l \right) \theta\left( a-b\sum_{l=1}^N x_l \right) - c \right] \left[ x_j + \frac{\left| \chi \right|}{N-1} \sum_{k\neq j} x_k \right] - \left( \left| \chi \right|+1 \right)\kappa.
 \label{eq:B_negative_mod}
\end{align}

When $\kappa > - \left| \chi \right| cx_\mathrm{max}/\left( 1 + \left| \chi \right| \right)$ with any $\chi<0$, we find
\begin{align}
 B \left( x_j, (N-1)x_\mathrm{max} \right) &= -c \left[ x_j + \left| \chi \right| x_\mathrm{max} \right] - \left( \left| \chi \right| +1 \right) \kappa \nonumber \\
 &< -c \left[ x_j + \left| \chi \right| x_\mathrm{max} \right] + \left| \chi \right| cx_\mathrm{max} \nonumber \\
 &= - c x_j \nonumber \\
 &\leq 0
\end{align}
for all $x_j$.
Therefore, $\overline{x}_j$ does not exist.

When $\left| \chi \right| > 1$ and $-cx_\mathrm{max} \leq \kappa \leq - \left| \chi \right| cx_\mathrm{max}/\left( 1 + \left| \chi \right| \right)$, we find
\begin{align}
 B \left( x_j, 0 \right) &= \left[ (a-bx_j) \theta(a-bx_j) - c \right] x_j + \left( \left| \chi \right| +1 \right) \left| \kappa \right| \nonumber \\
 &\geq -cx_\mathrm{max} + \left( \left| \chi \right| +1 \right) \left| \kappa \right| \nonumber \\
 &\geq -cx_\mathrm{max} +  \left| \chi \right| cx_\mathrm{max} \nonumber \\
 &> 0
\end{align}
for all $x_j$.
Therefore, $\underline{x}_j$ does not exist.

When $\left| \chi \right| \leq 1$ and $-cx_\mathrm{max}/\left( 1 + \left| \chi \right| \right) \leq \kappa \leq - \left| \chi \right| cx_\mathrm{max}/\left( 1 + \left| \chi \right| \right)$, we find
\begin{align}
 B \left( 0, x_{-j} \right) &= \left[ \left(a-b\sum_{k\neq j} x_k \right) \theta\left( a-b\sum_{k \neq j} x_k \right) - c \right] \frac{\left| \chi \right|}{N-1} \sum_{k\neq j} x_k + \left( \left| \chi \right| +1 \right) \left| \kappa \right| \nonumber \\
 &\geq - \left| \chi \right| cx_\mathrm{max} + \left( \left| \chi \right| +1 \right) \left| \kappa \right| \nonumber \\
 &\geq 0.
\end{align}
Therefore, we can use $x_j=0$ as $\overline{x}_j$.
Furthermore, we also find
\begin{align}
 B \left( x_\mathrm{max}, x_{-j} \right) &= -cx_\mathrm{max} -c \frac{\left| \chi \right|}{N-1} \sum_{k\neq j} x_k + \left( \left| \chi \right| +1 \right) \left| \kappa \right| \nonumber \\
 &\leq -cx_\mathrm{max} + \left( \left| \chi \right| +1 \right) \left| \kappa \right| \nonumber \\
 &\leq 0.
\end{align}
Therefore, we can use $x_j=x_\mathrm{max}$ as $\underline{x}_j$.
According to Proposition \ref{prop:two-point}, two-point negatively correlated autocratic strategies exist for this parameter region.

Finally, when $\left| \chi \right| \leq 1$ and $-cx_\mathrm{max} \leq \kappa < - cx_\mathrm{max}/\left( 1 + \left| \chi \right| \right)$, we find
\begin{align}
 B \left( x_j, 0 \right) &= \left[ (a-bx_j) \theta(a-bx_j) - c \right] x_j + \left( \left| \chi \right| +1 \right) \left| \kappa \right| \nonumber \\
 &\geq -cx_\mathrm{max} + \left( \left| \chi \right| +1 \right) \left| \kappa \right| \nonumber \\
 &> -cx_\mathrm{max} + cx_\mathrm{max} \nonumber \\
 &= 0
\end{align}
for all $x_j$.
Therefore, $\underline{x}_j$ does not exist.

\section*{Supporting information}
\paragraph*{S1 Appendix.}
\label{S1_Appendix}
{\bf Appendix.}
This appendix contains four sections, ``Calculation of expected payoffs when the opponents use $x^{(\mathrm{C})}$'', ``Additional numerical results'', ``Numerical results for only two adaptive agents'', and ``Fair zero-determinant alliance''.

\nolinenumbers


%
%
%

\begin{thebibliography}{10}

\bibitem{FudTir1991}
Fudenberg D, Tirole J.
\newblock Game Theory.
\newblock Massachusetts: MIT Press; 1991.

\bibitem{OsbRub1994}
Osborne MJ, Rubinstein A.
\newblock A Course in Game Theory.
\newblock Massachusetts: MIT Press; 1994.

\bibitem{Gib1992}
Gibbons R.
\newblock Game Theory for Applied Economists.
\newblock Princeton University Press; 1992.

\bibitem{Now2006}
Nowak MA.
\newblock Five rules for the evolution of cooperation.
\newblock Science. 2006;314(5805):1560--1563.

\bibitem{PreDys2012}
Press WH, Dyson FJ.
\newblock Iterated Prisoner{\textquoteright}s Dilemma contains strategies that
  dominate any evolutionary opponent.
\newblock Proceedings of the National Academy of Sciences.
  2012;109(26):10409--10413.

\bibitem{BNS1997}
Boerlijst MC, Nowak MA, Sigmund K.
\newblock Equal pay for all prisoners.
\newblock The American Mathematical Monthly. 1997;104(4):303--305.

\bibitem{StePlo2013}
Stewart AJ, Plotkin JB.
\newblock From extortion to generosity, evolution in the Iterated
  Prisoner{\textquoteright}s Dilemma.
\newblock Proceedings of the National Academy of Sciences.
  2013;110(38):15348--15353.

\bibitem{Aki2016}
Akin E.
\newblock The iterated prisoner's dilemma: good strategies and their dynamics.
\newblock Ergodic Theory, Advances in Dynamical Systems. 2016; p. 77--107.

\bibitem{HWTN2014}
Hilbe C, Wu B, Traulsen A, Nowak MA.
\newblock Cooperation and control in multiplayer social dilemmas.
\newblock Proceedings of the National Academy of Sciences.
  2014;111(46):16425--16430.

\bibitem{PHRT2015}
Pan L, Hao D, Rong Z, Zhou T.
\newblock Zero-determinant strategies in iterated public goods game.
\newblock Scientific Reports. 2015;5:13096.

\bibitem{McAHau2016}
McAvoy A, Hauert C.
\newblock Autocratic strategies for iterated games with arbitrary action
  spaces.
\newblock Proceedings of the National Academy of Sciences.
  2016;113(13):3573--3578.

\bibitem{TahGho2020}
Taha MA, Ghoneim A.
\newblock Zero-determinant strategies in repeated asymmetric games.
\newblock Applied Mathematics and Computation. 2020;369:124862.

\bibitem{KTZ2024}
Kang K, Tian J, Zhang B.
\newblock Cooperation and control in asymmetric repeated games.
\newblock Applied Mathematics and Computation. 2024;470:128589.

\bibitem{Ued2022}
Ueda M.
\newblock Unbeatable Tit-for-Tat as a Zero-Determinant Strategy.
\newblock Journal of the Physical Society of Japan. 2022;91(5):054804.

\bibitem{Ued2025}
Ueda M.
\newblock On the implementation of zero-determinant strategies in repeated
  games.
\newblock Applied Mathematics and Computation. 2025;489:129179.

\bibitem{Ued2022b}
Ueda M.
\newblock Necessary and Sufficient Condition for the Existence of
  Zero-Determinant Strategies in Repeated Games.
\newblock Journal of the Physical Society of Japan. 2022;91(8):084801.

\bibitem{Ued2023}
Ueda M.
\newblock Unexploitable games and unbeatable strategies.
\newblock IEEE Access. 2023;11:5062--5068.

\bibitem{HTS2015}
Hilbe C, Traulsen A, Sigmund K.
\newblock Partners or rivals? Strategies for the iterated prisoner's dilemma.
\newblock Games and Economic Behavior. 2015;92:41--52.

\bibitem{IchMas2018}
Ichinose G, Masuda N.
\newblock Zero-determinant strategies in finitely repeated games.
\newblock Journal of Theoretical Biology. 2018;438:61--77.

\bibitem{GovCao2020}
Govaert A, Cao M.
\newblock Zero-determinant strategies in repeated multiplayer social dilemmas
  with discounted payoffs.
\newblock IEEE Transactions on Automatic Control. 2020;66(10):4575--4588.

\bibitem{HRZ2015}
Hao D, Rong Z, Zhou T.
\newblock Extortion under uncertainty: Zero-determinant strategies in noisy
  games.
\newblock Phys Rev E. 2015;91:052803.

\bibitem{MamIch2020}
Mamiya A, Ichinose G.
\newblock Zero-determinant strategies under observation errors in repeated
  games.
\newblock Phys Rev E. 2020;102:032115.

\bibitem{MMI2021}
Mamiya A, Miyagawa D, Ichinose G.
\newblock Conditions for the existence of zero-determinant strategies under
  observation errors in repeated games.
\newblock Journal of Theoretical Biology. 2021;526:110810.

\bibitem{McAHau2017}
McAvoy A, Hauert C.
\newblock Autocratic strategies for alternating games.
\newblock Theoretical Population Biology. 2017;113:13--22.

\bibitem{DRWW2021}
Deng C, Rong Z, Wang L, Wang X.
\newblock Modeling replicator dynamics in stochastic games using Markov chain
  method.
\newblock In: Proceedings of the 20th International Conference on Autonomous
  Agents and Multiagent Systems; 2021. p. 420--428.

\bibitem{LiuWu2022}
Liu F, Wu B.
\newblock Environmental quality and population welfare in Markovian
  eco-evolutionary dynamics.
\newblock Applied Mathematics and Computation. 2022;431:127309.

\bibitem{HNS2013}
Hilbe C, Nowak MA, Sigmund K.
\newblock Evolution of extortion in Iterated Prisoner{\textquoteright}s Dilemma
  games.
\newblock Proceedings of the National Academy of Sciences.
  2013;110(17):6913--6918.

\bibitem{AdaHin2013}
Adami C, Hintze A.
\newblock Evolutionary instability of zero-determinant strategies demonstrates
  that winning is not everything.
\newblock Nature Communications. 2013;4(1):2193.

\bibitem{HNT2013}
Hilbe C, Nowak MA, Traulsen A.
\newblock Adaptive dynamics of extortion and compliance.
\newblock PLOS ONE. 2013;8(11):e77886.

\bibitem{SzoPer2014}
Szolnoki A, Perc M.
\newblock Evolution of extortion in structured populations.
\newblock Physical Review E. 2014;89(2):022804.

\bibitem{HWTN2015}
Hilbe C, Wu B, Traulsen A, Nowak MA.
\newblock Evolutionary performance of zero-determinant strategies in
  multiplayer games.
\newblock Journal of Theoretical Biology. 2015;374:115--124.

\bibitem{CWF2022}
Chen X, Wang L, Fu F.
\newblock The intricate geometry of zero-determinant strategies underlying
  evolutionary adaptation from extortion to generosity.
\newblock New Journal of Physics. 2022;24(10):103001.

\bibitem{CheFu2023}
Chen X, Fu F.
\newblock Outlearning extortioners: unbending strategies can foster reciprocal
  fairness and cooperation.
\newblock PNAS Nexus. 2023;2(6):pgad176.

\bibitem{BisNai2000}
Bischi GI, Naimzada A.
\newblock Global analysis of a dynamic duopoly game with bounded rationality.
\newblock In: Advances in Dynamic Games and Applications. Springer; 2000. p.
  361--385.

\bibitem{CheZin2014}
Chen J, Zinger A.
\newblock The robustness of zero-determinant strategies in iterated prisoner's
  dilemma games.
\newblock Journal of Theoretical Biology. 2014;357:46--54.

\bibitem{MMI2022}
Miyagawa D, Mamiya A, Ichinose G.
\newblock Adapting paths against zero-determinant strategies in repeated
  prisoner's dilemma games.
\newblock Journal of Theoretical Biology. 2022;549:111211.

\bibitem{MMHet2025}
McAvoy A, Madhushani~Sehwag U, Hilbe C, Chatterjee K, Barfuss W, Su Q, et~al.
\newblock Unilateral incentive alignment in two-agent stochastic games.
\newblock Proceedings of the National Academy of Sciences.
  2025;122(25):e2319927121.

\bibitem{DOS2012b}
Duersch P, Oechssler J, Schipper BC.
\newblock Unbeatable imitation.
\newblock Games and Economic Behavior. 2012;76(1):88--96.

\bibitem{DOS2014}
Duersch P, Oechssler J, Schipper BC.
\newblock When is tit-for-tat unbeatable?
\newblock International Journal of Game Theory. 2014;43(1):25--36.

\bibitem{Veg1997}
Vega-Redondo F.
\newblock The evolution of Walrasian behavior.
\newblock Econometrica: Journal of the Econometric Society.
  1997;65(2):375--384.

\bibitem{Ued2021b}
Ueda M.
\newblock Memory-two zero-determinant strategies in repeated games.
\newblock Royal Society Open Science. 2021;8(5):202186.

\bibitem{HSCN2018}
Hilbe C, {\v{S}}imsa {\v{S}}, Chatterjee K, Nowak MA.
\newblock Evolution of cooperation in stochastic games.
\newblock Nature. 2018;559(7713):246--249.

\bibitem{LiHao2019}
Li K, Hao D.
\newblock Cooperation enforcement and collusion resistance in repeated public
  goods games.
\newblock In: Proceedings of the AAAI Conference on Artificial Intelligence.
  vol.~33; 2019. p. 2085--2092.

\bibitem{HCN2018}
Hilbe C, Chatterjee K, Nowak MA.
\newblock Partners and rivals in direct reciprocity.
\newblock Nature Human Behaviour. 2018;2(7):469.

\end{thebibliography}

\end{document}